\documentclass[11pt,a4paper]{article}

\usepackage{amsthm}
\usepackage{slashed,makecell}

\usepackage{jheppub}

\newcommand{\be}{\begin{eqnarray}}
\newcommand{\ee}{\end{eqnarray}}

\newcommand{\eeq}{\end{equation}}
\newcommand{\beq}{\begin{equation}}

\allowdisplaybreaks \numberwithin{equation}{section}

\DeclareSymbolFont{AMSa}{U}{msa}{m}{n}
\DeclareSymbolFont{AMSb}{U}{msb}{m}{n}
\DeclareMathSymbol{\fieldR}{\mathalpha}{AMSb}{"52}

\DeclareMathOperator{\rank}{rank}

\newcommand{\CH}{\mathcal{H}}

\newcommand{\CN}{\mathcal{N}}

\DeclareMathOperator{\Tr}{Tr}

\newcommand{\NN}{\mathbb{N}}
\newcommand{\ZZ}{\mathbb{Z}}
\newcommand{\RR}{\mathbb{R}}
\newcommand{\CC}{\mathbb{C}}
\newcommand{\HH}{\mathbb{H}}
\newcommand{\g}{\mathfrak{g}}
\newcommand{\h}{\mathfrak{h}}
\newcommand{\Pg}{\mathfrak{P}}
\newcommand{\Qg}{\mathfrak{Q}}
\newcommand{\hreg}{\mathsf{h}}

\newcommand{\rootQ}{\mathit{Q}}

\newcommand{\dcox}{h^\vee} 

\newcommand{\ex}{\operatorname{e}}

\def\beq{\begin{equation}}
\def\eeq{\end{equation}}
\def\bea{\begin{eqnarray}}
\def\eea{\end{eqnarray}}

\def\<{\langle}

\newtheorem{theorem}{Theorem}
\newtheorem{lemma}{Lemma}
\newtheorem{proposition}{Proposition}

\title{Fun with \boldmath{$F_{24}$}}
\author[1, 2]{Sarah M. Harrison}
\author[3, 4]{Natalie M. Paquette}
\author[5]{Daniel Persson}
\author[6]{Roberto Volpato}

 \affiliation[1]{\small Department of Mathematics and Statistics, McGill University, Montreal, QC, Canada}
\affiliation[2]{\small Department of Physics, McGill University, Montreal, QC, Canada}
\affiliation[3]{\small Walter Burke Institute for Theoretical Physics, California Institute of Technology,
Pasadena, CA 91125, USA}
\affiliation[4]{\small School of Natural Sciences, Institute for Advanced Study, Princeton, NJ, 08540, USA}
\affiliation[5]{Department of Mathematical Sciences, Chalmers University of Technology, Gothenburg, Sweden}
\affiliation[6]{\small Dipartimento di Fisica e Astronomia `Galileo Galilei', Universit\`a di Padova \& INFN, sez. di Padova, Via Marzolo 8, 35131, Padova, Italy}
\emailAdd{sarharr@physics.mcgill.ca}
\emailAdd{paquette@ias.edu}
\emailAdd{danper@chalmers.se}
\emailAdd{volpato@pd.infn.it}

 \abstract{\noindent We study some special features of $F_{24}$, the holomorphic $c=12$ superconformal field theory (SCFT) given by 24 chiral free fermions. We construct eight different Lie superalgebras of ``physical'' states of a chiral superstring compactified on $F_{24}$, and we prove that they all have the structure of Borcherds-Kac-Moody superalgebras. This produces a family of new examples of such superalgebras. The models depend on the choice of an $\CN=1$ supercurrent on $F_{24}$, with the admissible choices labeled by the semisimple Lie algebras of dimension 24. We also discuss how $F_{24}$, with any such choice of supercurrent, can be obtained via orbifolding from another distinguished $c=12$ holomorphic SCFT, the $\CN=1$  supersymmetric version of the chiral CFT based on the $E_8$ lattice.  }

\arxivnumber{2009.14710}
\begin{document}
\maketitle


\section{Introduction}

Quantum field theories of two dimensional fermions are among the simplest to write down, and nevertheless have remarkably rich physics. For example, perhaps the simplest 2d conformal field theory (CFT), the critical point of the 2d Ising model, can be described in terms of a free Majorana fermion. More elaborate fermionic CFTs appear as edge modes in quantum Hall systems \cite{Tong:2016kpv}, as well as in the classification of symmetry-protected topological (SPT) phases \cite{Senthil:2014ooa}. Two-dimensional Bose-Fermi duality also relates fermionic CFTs to lattice CFTs of bosons, which appear in, e.g., toroidal compactifications of string theory. Recent works on dualities have also shed light on subtle discrete invariants required to understand the rich physics of two dimensional fermions \cite{KTTW, KT, KarchTongTurner}. 

In this paper, we will rather emphasize intricate Lie algebraic structures hidden in the deceptively simple physics of 2d free fermions \footnote{Interesting work in a similar spirit appeared recently in \cite{Troost}, which studied quantum mechanical fermions valued in (gauged) Lie algebras.}. As has been emphasized in several corners of the mathematical physics landscape (see, to give an incomplete list, \cite{GaiottoJFWitten, GaiottoJF, TongTurner} and references therein), systems with certain distinguished numbers of fermions can enjoy special properties and intertwine with several species of modular objects; our interest will be in various symmetry structures present in a system of 24 chiral fermions and some associated automorphic forms.

In this paper we study a system of 24 free chiral fermions in 2d. This is a holomorphic superconformal field theory, or super-vertex operator algebra (SVOA), with central charge 12 and which we refer to throughout as $F_{24}$. It is notable because it is one of three so-called self-dual SVOAs with central charge 12.\footnote{A self-dual SVOA W is one that is rational and the unique irreducible W-module (up to isomorphism.)} These theories were classified in \cite{CDR} and are given by (up to isomorphism):
\begin{enumerate}
	\item $V^{fE_8}$: This is the theory of 8 chiral bosons compactified on $\mathbb R^8/\Lambda_{E_8}$, where $\Lambda_{E_8}$ is the $E_8$ root lattice, and their 8 fermionic superpartners. 
	\item $V^{f\natural}$: This is the unique holomorphic SCFT with $c=12$ and no weight-1/2 fields. First discussed in \cite{Duncan}, it has a unique $\mathcal N=1$ superconformal structure which is stabilized by Conway's largest sporadic group.
	\item $F_{24}$: This is a theory of 24 free chiral fermions. One can build an $\mathcal N=1$ superconformal structure by taking a linear combination of cubic Fermi terms, and the allowed choices are classified by semisimple Lie algebras of dimension 24. Each of these generates an affine Kac-Moody algebra, of which there are eight possibilities:
	\begin{align}(\widehat{su}(2)_2)^{\oplus 8}\ ,\quad (\widehat{su}(3)_3)^{\oplus 3}\ ,\quad \widehat{su}(4)_4\oplus (\widehat{su}(2)_2)^{\oplus 3}\ ,\quad \widehat{su}(5)_5\ ,\quad  \widehat{so}(5)_3\oplus \hat g_{2,4}\ ,\notag\\
\widehat{so}(5)_3\oplus \widehat{su}(3)_3\oplus (\widehat{su}(2)_2)^{\oplus 2}\ ,\quad \widehat{so}(7)_5\oplus \widehat{su}(2)_2\ ,\quad \widehat{sp}(6)_4\oplus \widehat{su}(2)_2\ ,\notag
\end{align}
	which we describe in \S \ref{s:F24}.
\end{enumerate}

On the face of it, these three theories are quite different--they have notably different constructions and symmetry groups. However, as is described in  \cite{Duncan, Duncan:2014eha,CDR,ACDV} and \S \ref{s:F24fromOb}, by gauging symmetries, one can move from one to the other. Furthermore, in \cite{Sch1} and  \cite{Harrison:2018joy}, respectively, the theories $V^{fE_8}$ and $V^{f\natural}$ have been used to furnish constructions of a special type of infinite-dimensional Lie superalgebra known as a Borcherds-Kac-Moody (BKM) superalgebra. BKM algebras were originally introduced by Borcherds \cite{BorcherdsMM} in his proof of the monstrous moonshine conjectures of Conway and Norton \cite{CN}, and Thompson \cite{Thompson1,Thompson2}.\footnote{For reviews of moonshine, see, e.g.,   \cite{Duncan:2014vfa,Anagiannis:2018jqf}.} The monster BKM arises from BRST quantization of a chiral bosonic string theory and elucidates connections between modular functions, the monster sporadic simple group, and the physics of 2d CFT.

One of the goals of the present paper is to describe the construction of a family of BKM superalgebras based on the theory $F_{24}$,  similar to the constructions of BKM superalgebras based on $V^{fE_8}$ and $V^{f\natural}$ mentioned above. For each choice of $\mathcal N=1$ superconformal  structure for $F_{24}$ we construct a corresponding BKM superalgebra $\g$ with Kac-Moody symmetry determined by the choice of $\mathcal N=1$ supercurrent. The main results of this work are threefold: 
\begin{itemize}
\item We show that all choices of $\CN=1$ structure on $F_{24}$ can be obtained from orbifolds of the SVOA $V^{fE_8}$ (see \S\ref{s:F24fromOb}),
\item We prove that the Lie superalgebra $\g$ satisfies the conditions of a BKM superalgebra (see Theorem \ref{mainthm} in \S \ref{sec:LiePhysStates}),
\item We provide an infinite product formula for the Borcherds-Weyl-Kac denominator for each $\g$ (see \S\ref{sec:denom}). 
\end{itemize}

Besides the fact that we construct a new family of examples of  (super)-BKM algebras, of which there are only very few explicit constructions, one of our long-term interests is to elucidate the connection between BKM algebras and BPS states in string theory,  which was originally envisaged by Harvey and Moore \cite{HM1,HM2}. They suggested that BPS states in string and field theories with extended supersymmetry should form an algebra, and that---at least in some contexts---this algebra would be a generalized\footnote{That is, a Borcherds-Kac-Moody algebra.} Kac-Moody algebra (or contain one as a subalgebra.) 

An interesting example of this proposal, similar in spirit to the present study, was studied by the last three named authors \cite{PPV1, PPV2}, where it was  found that \emph{spacetime} BPS states in a second quantized heterotic string theory furnished a natural module over the Monster BKM.\footnote{See also \cite{Cheng:2008fc} for a proposal for the appearance of a BKM algebra in  string theory in a quite different context, building on the pioneering results of \cite{DVV}.} The worldsheet string theory for this construction employed the Monster vertex operator algebra $V^{\natural}$, and the construction has been used to shed light on the physical interpretation of the genus zero property of monstrous moonshine. 
Similarly, the theories  $V^{fE_8},V^{f\natural},F_{24}$ all naturally occur as (chiral halves) of worldsheet CFTs at special points in the moduli space of maximally supersymmetric type II string compactifications to 2d. We expect that the BKM superalgebras constructed in \cite{Sch1}, \cite{Harrison:2018joy}, and this paper, occur as algebras acting on spacetime BPS states at such special points. This is a notion we make precise in upcoming work \cite{HPPV}. Furthermore, given the close relation between $V^{fE_8},V^{fE_8},F_{24}$ via orbifolding, by analogy with the constructions in \cite{Persson:2015jka,Persson:2017lkn,PPV1} we expect that we may uncover new 2d spacetime string dualities by considering worldsheet theories which are a tensor product $V \otimes \bar W$ with $V,W$ taken to be one of these theories. 

The outline of the rest of the paper is as follows. In \S \ref{s:F24}, we review the construction of $F_{24}$, its canonically--twisted module $F_{24}^{tw}$, the allowed choices of $\mathcal N=1$ supercurrent, and its character. In \S \ref{s:F24fromOb}, we describe how $F_{24}$ with a choice of $\mathcal N=1$ structure can be obtained from orbifolds of the SVOA $V^{fE_8}$. In \S\ref{s:BKM} we explain the general construction of BKM superalgebras from $\mathcal{N}=1$ SVOAs. We then go on to construct a family of BKM superalgebras $\mathfrak{g}$ from $F_{24}$ with a choice of $\mathcal N=1$ superconformal structure. In \S\ref{sec:BKMproof} we prove our main theorem, showing that $\mathfrak{g}$ is a super BKM-algebra. In this section we also discuss the denominator and super denominator formulas of $\mathfrak{g}$. We give more details on the example with $\hat A_1^8$ Kac-Moody symmetry in the following \S \ref{sec:A1^8}. We conclude with a brief discussion of open questions in \S \ref{sec:discussion}. Finally, two appendices provide further details about multivariable Jacobi forms (\S \ref{a:multiJacobi}) and the relative BRST cohomology for physical states in our theories (\S \ref{a:details}).

\section{The SVOA and its $\CN=1$ structures}\label{s:F24}

\subsection{Construction}
The starting point of our construction is a simple holomorphic chiral superconformal field theory (SCFT) $F_{24}$ of central charge $c=12$, given by $24$ chiral free fermions. In mathematical language, this is a self-dual $C_2$-cofinite  super vertex operator algebra (SVOA) of CFT type of central charge $12$. In our definition of SCFT (or SVOA) we do not include the choice of an $\CN=1$ subalgebra. For this reason, we refer to $F_{24}$ as a single SVOA, even though, as discussed below, it admits different $\CN=1$ structures.

This theory is generated by $24$ chiral free fermions, $\lambda^1(z),\ldots,\lambda^{24}(z)$, with OPE
\be \lambda^i(z)\lambda^j(w)\sim \frac{\delta^{ij}}{z-w},
\ee
and stress-energy tensor,
\be\label{stress} T(z)=-\frac{1}{2}\sum_{i=1}^{24} :\!\lambda^i\partial \lambda^i\!:\!(z)\ ,
\ee with respect to which $\lambda^i$ have conformal weight $1/2$. The $\binom{24}{2}=276$ currents $\lambda_i\lambda_j$, $1\le i<j\le 24$, generate an $\widehat{so}(24)_1$ Kac-Moody algebra, which is  the bosonic (even) subVOA of $F_{24}$. By bosonization, the same SVOA can be described as a lattice model based on the odd unimodular lattice $\ZZ^{12}$, with the $\widehat{so}(24)_1$ algebra corresponding to the $D_{12}$ sublattice. The space of odd (fermionic) states in $F_{24}$ forms a vector module for the $\widehat {so}(24)_1$ algebra.

$\CN=1$ structures in free fermion theories were classified in \cite{Goddard:1984hg}. An $\CN=1$ supercurrent must be a linear combination of conformal primaries of weight $3/2$, so it must be of the form
\be\label{freefermG} G= -\frac{i}{6}\sum_{i,j,k} c_{ijk} :\lambda^i\lambda^j\lambda^k:\ 
\ee for some totally antisymmetric $c_{ijk}\in \CC$. The stress tensor \eqref{stress} and the supercurrent $G(z)$ defined by \eqref{freefermG} generate an $\CN=1$ superconformal algebra at central charge $12$ if and only if the following conditions are satisfied \cite{Goddard:1984hg}:
\be\label{algJacobi} \sum_k(c_{ijk}c_{klm}+c_{lik}c_{kjm}+c_{jlk}c_{kim})=0\ ,
\ee
\be\label{algnorm}  \sum_{k,l} c_{ikl}c_{jkl}=2\delta_{ij}\ .
\ee The first condition is equivalent to the requirement that the $G(z)G(0)$ OPE does not contain any singular term with four fermions $\lambda\lambda\lambda\lambda$; the second is equivalent to the requirement that the $z^{-1}$ term in the OPE reproduces the stress-energy tensor $T(z)$. 

These conditions imply that $c_{ijk}$ are the structure constants of a semisimple Lie algebra $g$ of dimension $24$ (and any rank), i.e. $g$ is the complex Lie algebra generated by $t^1,\ldots, t^{24}$ with commutation relations
\be [t^j,t^k]=ic_{jkl} t^l\ .
\ee
Given a choice of $c_{ijk}$ satisfying \eqref{algJacobi} and \eqref{algnorm}, the $24$ currents
\be J^a(z)=-\frac{i}{2}\sum_{j,k=1}^n c_{ajk} :\lambda_j\lambda_k:(z)\ ,\qquad a=1,\ldots, 24\ ,
\ee satisfy the OPE 
\be\label{currentOPE} J^a(z)J^b(0)=\frac{\delta_{ab}}{z^2}+\frac{i}{z}\sum_{k=1}^n c_{abk} J^k(0)+\ldots,
\ee
 which shows that the $J^a(z)$ generate an affine Kac-Moody algebra $\hat g$ based on the finite Lie algebra $g$. Thus, $g$ is a $24$-dimensional subalgebra of the $so(24)$ algebra generated by the zero modes of fermion bilinears $\lambda^i\lambda^j$.
The OPEs
\be G(z)\lambda^a(0)=\frac{J^a(0)}{z}+\ldots, ~~~{\rm and}
\ee
\be J^a(z)\lambda^b(0)=\frac{i}{z}\sum_{k=1}^nc_{abk}\lambda^k(0)+\ldots
\ee
show that the currents $J^a(z)$ are singled out as the $\CN=1$ descendants of the $24$ free fermions $\lambda^a$, while the remaining $252$ currents are superconformal primaries. Furthermore, the free fermions $\lambda^a$ transform in the adjoint representation with respect to the algebra $g$ of zero modes.

The SVOA admits a canonical non-degenerate invariant bilinear form $(\cdot |\cdot)$, given by the Zamolodchikov metric. By \eqref{currentOPE}, the set of currents $\{J^a\}_{a=1,\ldots,24}$ is orthonormal with respect to this bilinear form.  In the following sections, we will need to choose a  normalization for the Cartan-Killing form $(\cdot |\cdot)_g$ on the finite dimensional Lie algebra $g$. It is convenient to choose
 \be\label{Killing} (t|u)_g=\frac{1}{2}\Tr({\rm Ad}(t){\rm Ad}(u))\ .
 \ee With this choice, the Cartan-Killing form $(\cdot|\cdot)_g$ on the algebra of zero modes $J^a_0$ coincides with the bilinear form induced by the Zamolodchikov metric, since
 \be (J^a_0|J^b_0)_g=\frac{1}{2}\Tr({\rm Ad}(J^a_0){\rm Ad}(J^b_0))=-\frac{1}{2}\sum_{j,k=1}^n c_{ajk}c_{bkj}=\delta_{ab}\ ,
 \ee where we used \eqref{algnorm} and that ${\rm Ad}(J^a_0)_{jk}=-ic_{ajk}$ in the basis $\{J^a_0\}$. In the following, we will often drop the subscript $g$ on the Killing form. Notice that if $g$ is a direct sum of simple components $g=\oplus_i g_i$, where $g_i$ has  dual Coxeter number $\dcox_{g_i}$, then the restriction of the Killing form to $g_i$ is such that the long roots have square length $2/\dcox_{g_i}$. 
 
In terms of modes $J^a(z)=\sum_{n\in \ZZ} J^a_n z^{-n-1}$, the OPE \eqref{currentOPE} yields the commutation relations
\be [J^a_n,J^b_m]= n\delta^{ab}\delta_{m,-n}+i\sum_c c_{abc}J^c_{m+n}\ .
\ee We recall that, for an affine Kac-Moody algebra $\hat g'_k$ based on a \emph{simple} algebra $g'$ at level $k$, the commutation relations read
\be [t_n,u_m]=k\frac{(t|u)_{g'}}{\dcox_{g'}}\delta^{ab}\delta_{m,-n}+({\rm Ad}(t).u)_{n+m}\ ,
\ee  when the Killing form is normalized as in \eqref{Killing}.\footnote{More generally, the relations are  $$[t_n,u_m]=k\frac{|\theta|^2}{2}(t|u)\delta^{ab}\delta_{m,-n}+({\rm Ad}(t).u)_{n+m}\ ,$$ where $|\theta|^2$ is the length of the long roots. With the choice \eqref{Killing} for the normalization, one has $|\theta|^2=\frac{2}{\dcox_{g'}}$, hence the formula.} Comparing these equations, we see that if $g$ is the sum $g=\oplus_i g_i$ of simple components of dual Coxeter number $\dcox_{g_i}$, the affine algebra $\hat g$ is given by 
\be \hat g=\oplus_i (\hat g_i)_{\dcox_{g_i}}\ ,\ee 
i.e. the levels of the simple components equal the dual Coxeter numbers.\footnote{The same conclusion can be reached by noticing that, for a simple algebra $g'$, the dual Coxeter number is the embedding index of $g'\subset so(\dim g')$. The embedding index is the ratio of the levels for the corresponding embedding of affine algebras. Since $\dim g'$ fermions generate an algebra $\hat{so}(\dim g')_1$ at level $1$, we have that $\hat g'$ must have level $\dcox_{g'}$.}

 For $\dim g=24$, there are eight distinct possibilities for $g$:
\be\label{N1algebras} A_1^8\ ,\quad A_2^3\ ,\quad A_3A_1^3\ ,\quad A_4\ ,\quad  B_2G_2\ ,\quad B_2A_2A_1^2\ ,\quad B_3A_1\ ,\quad C_3A_1\ .
\ee
The corresponding affine algebras are,
\be\label{N1algebras3} \hat A_{1,2}^8\, ,\quad \hat A_{2,3}^3\, ,\quad \hat A_{3,4}\hat A_{1,2}^3\, ,\quad \hat A_{4,5}\, ,\quad  \hat B_{2,3}\hat G_{2,4}\, ,\quad \hat B_{2,3}\hat A_{2,3}\hat A_{1,2}^2\, ,\quad \hat B_{3,5}\hat A_{1,2}\, ,\quad \hat C_{3,4}\hat A_{1,2}\, ,
\ee that is,
\begin{align}\label{N1algebras2} (\widehat{su}(2)_2)^{\oplus 8}\ ,\quad (\widehat{su}(3)_3)^{\oplus 3}\ ,\quad \widehat{su}(4)_4\oplus (\widehat{su}(2)_2)^{\oplus 3}\ ,\quad \widehat{su}(5)_5\ ,\quad  \widehat{so}(5)_3\oplus \hat g_{2,4}\ ,\\
\widehat{so}(5)_3\oplus \widehat{su}(3)_3\oplus (\widehat{su}(2)_2)^{\oplus 2}\ ,\quad \widehat{so}(7)_5\oplus \widehat{su}(2)_2\ ,\quad \widehat{sp}(6)_4\oplus \widehat{su}(2)_2\ .\notag
\end{align}

Finally, this SVOA admits a unique (up to isomorphism) canonically twisted module $F_{24}^{tw}$, which also admits an invariant non-degenerate bilinear form $(\cdot |\cdot)$. Using the string theory terminology, we will often refer to the SVOA $F_{24}$ as the Neveu-Schwarz (NS) sector and to its twisted module as the Ramond (R) sector. Recall that the even subalgebra of $F_{24}$ is the bosonic lattice VOA  $V_{D_{12}}$ based on the $D_{12}$ lattice. This VOA $V_{D_{12}}$ has four irreducible modules which are in one-to-one correspondence with the cosets $D_{12}^*/D_{12}$. We can label the four modules as `adjoint', `vector', `spinor' and `conjugate spinor' in terms of their $so(24)$ representations. While $F_{24}$ is given by the direct sum of the adjoint and vector $V_{D_{12}}$-modules, the  canonically twisted module can be identified with the direct sum of the two spinor $V_{D_{12}}$-modules, with opposite fermion number. This description immediately shows that the lowest conformal weight in the Ramond sector is $3/2$, and in particular there are no states of weight $1/2$. For any choice of the $\CN=1$ supercurrent $G(z)$, the relation $G_0^2=L_0-\frac{1}{2}$ implies that the zero mode $G_0$ has zero kernel in the Ramond sector, and therefore establishes an isomorphism between the components with positive and negative fermion number.

\subsection{Partition functions}\label{s:partfunct}

In this section, we compute the partition functions of the SVOA $F_{24}$ (NS sector) and its canonically twisted module $F_{24}^{tw}$ (R sector).

First we describe our notation. Let us choose a Cartan subalgebra $h$ of the Lie algebra $g$ and let $g=g_-\oplus h\oplus g_+$ be a triangular decomposition.  Let  $\Delta^+\subset h^*\cong \CC^r$ be the set of positive roots, where $r=\rank(g)$ is the rank of $g$, $\alpha_1,\ldots, \alpha_r\in h^*$ be the simple roots, and $\alpha^\vee_1,\ldots, \alpha^\vee_r\in h$  be the coroots. 
 We normalize the Killing form $(\cdot|\cdot)_g$ as in \eqref{Killing} so that the long roots in each simple component $g'$ of $g$ have  length-squared  $2/\dcox_{g'}$, where $\dcox_{g'}$ is the dual Coxeter number of $g'$.  
 We denote by $\rootQ_g=\sum_i\ZZ \alpha_i\subset h^*$ the root lattice, by $\rootQ_g^\vee=\sum_i\ZZ \alpha_i^\vee \subset h$ the coroot lattice and by
\be P_g\equiv (\rootQ_g^\vee)^*=\{w\in h^*\mid w(\alpha_i^\vee)\in \ZZ \ ,\forall i\}
\ee its dual lattice (the weight lattice).\footnote{Occasionally, when it is clear from context which $g$ we are studying, we will omit the subscript on the root, coroot and weight lattices.} The Killing form \eqref{algnorm} defines an isomorphism $i:h\to h^*$, which we often keep implicit, simply identifying $h$ and $h^*$. With this Killing form, the coroot lattice $Q_g^\vee$ is even, so that $i(Q^\vee_g)\subset P_g$.\footnote{We notice that the isomorphism $i$ depends on the normalization of the Killing form, which is not completely standard (long roots do not have length $2$, but $2/\dcox_g$). For example, for $su(2)$, the root has length $2/2=1$, so that one can identify $Q_g=\ZZ$, $P_g=\frac{1}{2}\ZZ$ and $i(Q^\vee_g)=2\ZZ$.}
 
We can choose the basis vectors of the SVOA to be simultaneous eigenstates of $L_0$ and of the Cartan generators $(\alpha^\vee_1,\ldots,\alpha^\vee_r)$ of $g$. The $r$-tuple of eigenvalues (the charges) for $(\alpha^\vee_1,\ldots,\alpha^\vee_r)$ is a weight $w\equiv (w_1,\ldots,w_r)\in P_g$. To keep track of both the $L_0$ and $h$ eigenvalues, we introduce characters depending on $\tau\in \HH$ and on $\xi\in h\cong \CC^r$, namely:
\begin{align}
&\phi_{NS}(\tau,\xi)=Tr_{NS}(q^{L_0-\frac{c}{24}} e^{2\pi i \xi})=\sum_{n\in \frac{1}{2}\ZZ}\sum_{w\in P_g}c_{NS}(n,w)q^ne^{2\pi i ( \xi| w)}\ ,\\
&\phi_{\widetilde {NS}}(\tau,\xi)=Tr_{NS}(q^{L_0-\frac{c}{24}} e^{2\pi i \xi}(-1)^F)=\sum_{n\in \frac{1}{2}\ZZ}\sum_{w\in P_g}(-1)^{2n+1}c_{\widetilde{NS}}(n,w)q^ne^{2\pi i (\xi| w)}\ ,\\
&\phi_{R}(\tau,\xi)=Tr_{R}(q^{L_0-\frac{c}{24}}  e^{2\pi i \xi})=\sum_{n\in \ZZ}\sum_{w\in P_g}c_{R}(n,w)q^ne^{2\pi i ( \xi|w)}\ ,\\
&\phi_{\tilde R}(\tau,\xi)=Tr_{R}(q^{L_0-\frac{c}{24}}  e^{2\pi i \xi}(-1)^F)=\sum_{n\in \ZZ}\sum_{w\in P_g}c_{\tilde R}(n,w)q^ne^{2\pi i (\xi| w)}\ .
\end{align}
A direct calculation then gives (here $\rho=\frac{1}{2}\sum_{\alpha\in\Delta^+}\alpha\in  P_g$ denotes the Weyl vector of $g$):
\begin{align} \phi_{\widetilde{NS}}(\tau,\xi)&=q^{-1/2} \prod_{n=1}^\infty \left[(1-q^{n-\frac{1}{2}})^r\prod_{\alpha\in \Delta_+} (1-q^{n-\frac{1}{2}}e^{2\pi i (\xi|\alpha)})(1-q^{n-\frac{1}{2}}e^{-2\pi i (\xi|\alpha)})\right]\notag\\
&= \frac{\theta_4(\tau,0)^{r/2}\prod_{\alpha\in \Delta_+}\theta_4(\tau,(\xi|\alpha))}{\eta(\tau)^{12}}\ \end{align}
\begin{align} \phi_{NS}(\tau,\xi)&=q^{-1/2} \prod_{n=1}^\infty \left [(1+q^{n-\frac{1}{2}})^r\prod_{\alpha\in \Delta_+} (1+q^{n-\frac{1}{2}}e^{2\pi i (\xi|\alpha)})(1+q^{n-\frac{1}{2}}e^{-2\pi i (\xi|\alpha)})\right]\notag\\
&= \frac{\theta_3(\tau,0)^{r/2}\prod_{\alpha\in \Delta_+}\theta_3(\tau,(\xi|\alpha))}{\eta(\tau)^{12}}\ \end{align}
\begin{align} \phi_{R}(\tau,\xi)&=e^{-2\pi i (\xi|\rho)}2^\frac{r}{2}q\prod_{m=1}^\infty (1+q^m)^r\prod_{\alpha\in \Delta_+} \left [\prod_{n=0}^\infty(1+q^ne^{2\pi i (\xi|\alpha)})\prod_{k=1}^\infty(1+q^ke^{-2\pi i (\xi|\alpha)})\right]\notag\\
&= \frac{\theta_2(\tau,0)^{r/2}\prod_{\alpha\in \Delta_+}\theta_2(\tau,(\xi|\alpha))}{\eta(\tau)^{12}}=2^\frac{r}{2}\frac{\eta(2\tau)^r\prod_{\alpha\in \Delta_+}\theta_2(\tau,(\xi|\alpha))}{\eta(\tau)^{12+r/2}}\ ,\label{JacRam}\end{align}
and
\be \phi_{\tilde R}(\tau,\xi)=0\ .
\ee  See appendix \ref{a:multiJacobi} for the definition of the Jacobi theta functions. 

The last equality follows because, as discussed above, $\ker G_0=0$ so that the Ramond spaces with positive and negative fermion number are isomorphic. In this computation, we use the fact that the Ramond ground states form a $2^{12}$-dimensional representation of the algebra $g$, which is isomorphic to the direct sum of $2^{r/2}$ copies of the representation $L_\rho$ whose highest weight is the Weyl vector $\rho$. In particular, $\dim L_\rho=2^N$, with $N=(24-r)/2$ being the number of positive roots. To show this, we first notice that the space of ground states forms an irreducible module for the Clifford algebra of fermionic zero modes. Using this description, it is easy to check that the difference between the highest and the lowest weights in this representation is the sum over the positive roots, i.e. $2\rho$, and that the multiplicity of either the lowest or highest weight is $2^{r/2}$: the space of highest weight vectors is itself a module for the Clifford subalgebra of $r$ fermionic zero modes corresponding to the Cartan generators of $g$. Finally, this $g$-representation must be self-conjugate, because the canonically twisted module of $F_{24}$ is unique, so it must be isomorphic to its dual. Thus, the highest weight must be the opposite of the lowest, and therefore equal to $\rho$. 

In the following, we also need the linear combinations
\be\label{Jaceven} \phi_{NS\pm }(\tau,\xi)=Tr_{NS}(q^{L_0-\frac{c}{24}} e^{2\pi i \xi}\frac{1\pm (-1)^F}{2})=\frac{1}{2}(\phi_{\widetilde{NS}}(\tau,\xi)\pm \phi_{NS}(\tau,\xi))
\ee and
\be\label{Jacodd} \phi_{R\pm}(\tau,\xi)=Tr_{R}(q^{L_0-\frac{c}{24}} e^{2\pi i \xi}\frac{1\pm(-1)^F}{2})=\frac{1}{2}\phi_{R}(\tau,\xi)\ ,
\ee giving the partition functions on the  eigenspaces of the fermion number.

As shown in appendix \ref{a:multiJacobi}, these functions admit a Fourier expansion
\be\label{JacobiFourier} \phi_X(\tau,\xi)=\sum_{n}\sum_{w\in P_g}c_X(n,w)q^ne^{2\pi i ( \xi|w)}\ ,\ee
 where $X\in \{NS,\widetilde{NS},R,\tilde{R},NS\pm,R\pm\}$, and the sum over $n$ is over integers in the Ramond sector or when $X=NS-$, and over half-integers in all other cases. The sum over $w$ can be reduced to a sum over the root lattice $\rootQ_g\subset P_g$ in the NS sector, or over the coset $\rho+\rootQ_g\subset P_g$ in the R sector. More generally, the sum over $w$ can be restricted to
 \be\label{rooti} \tilde Q_g:=Q_g\cup (\rho+Q_g)\ .
 \ee

In appendix \ref{a:multiJacobi} we show that the coefficients $c_X(n,w)$ depend on $n$ and $w\in P_g$ only through the discriminant
\be D=2n-(w|w)_g
\ee and on the class $[w]$ of $w$ in $P_g/i(\rootQ^\vee_g)$; we will sometimes write $c_X(n,w)\equiv c_X(D,[w])$ to emphasize this dependence. 

In particular, when $X\in \{NS-,R\pm\}$, the coefficients $c_X(n,w)$ are non-zero only when
\be \label{coeffboundA} c_X(n,w)\neq 0\qquad \Rightarrow\qquad \begin{cases}n\ge 0\\ 2n-(w|w)_g\ge -m([w])\ ,\end{cases}
\ee where \be m([w])=\min\{(w'|w')\mid w'\in w+i(\rootQ^\vee_g)\}\ee
is the minimal square length of vectors in the coset $w+i(\rootQ^\vee_g)$. Since there are only a finite number of cosets in $P_g/i(\rootQ^\vee_g)$, we can also give a bound on the discriminant that is independent of the class of $w$,
\be \label{coeffboundB} c_X(n,w)\neq 0\qquad \Rightarrow\qquad \begin{cases}n\ge 0\\ 2n-(w|w)_g\ge -M\ ,\end{cases} 
\ee where 
\be\label{Mbound} M=\max_{[w]\in P_g/i(\rootQ^\vee_g)}m([w])\ .
\ee Finally, for $n=0$, one has
\be c_{NS-}(0,w)=\begin{cases}
	1 & \text{if $w$ is a root of $g$}\\
	r & \text{if $w=0$}\\
	0 & \text{otherwise}
\end{cases}
\ee while $c_{R\pm}(0,w)=0$ for all $w\in P_g$, given the absence of Ramond states of weight $1/2$.

The coefficients $c_{NS-}$ and $c_{R\pm}$ will correspond to the (respectively, even and odd) root multiplicities for the BKM superalgebra that we will construct in \S \ref{s:BKM}.

If we set all $\xi=0$, we get the same formulas for all choices of $\CN=1$ structure, namely
\be \phi_{ NS}(\tau,0)=\Tr_{NS}(q^{L_0-\frac{c}{24}})=-\frac{\eta(\frac{\tau+1}{2})^{24}}{\eta(\tau)^{24}}=q^{1/2}+24+276q^{1/2}+2048q+\ldots
\ee
\be \phi_{\widetilde{NS}}(\tau,0)=\Tr_{NS}((-1)^F q^{L_0-\frac{c}{24}})=\frac{\eta(\tau/2)^{24}}{\eta(\tau)^{24}}=q^{1/2}-24+276q^{1/2}-2048q+\ldots
\ee
\be \phi_{R}(\tau,0)=\Tr_{R}(q^{L_0-\frac{c}{24}})=2^{12}\frac{\eta(2\tau)^{24}}{\eta(\tau)^{24}}=4096q+98304q^2+\ldots
\ee
\be \phi_{\tilde R}(\tau,0)=\Tr_{R}((-1)^Fq^{L_0-\frac{c}{24}})=0\ .
\ee
In particular,
\be \phi_{NS-}(\tau,0)=24+2048q+49152q^2+\ldots
\ee and
\be \phi_{R+}(\tau,0)=\phi_{R-}(\tau,0)=2048q+49152q^2+\ldots
\ee showing that $\phi_{NS-}(\tau,0)=\phi_{R\pm}(\tau,0)+24$.

Having established some basic properties of the $F_{24}$ SVOA, we will next explain how one can obtain this theory, including the choice of $\CN=1$ structure, via orbifolds of $V^{fE_8}$. Though we do not undertake a full string theoretic construction in this work (though see \cite{HPPV}), the orbifolds discussed in the next section will be a precursor to various spacetime dualities relating string theories with different perturbative worldsheet descriptions based upon the $c=12$ self-dual SVOAs. 

\section{$F_{24}$ from orbifolds of $V^{fE_8}$}\label{s:F24fromOb}

In \cite{CDR} it was shown that the SVOA $F_{24}$ can be obtained from the SVOAs $V^{fE_8}$ or $V^{f\natural}$ by an orbifold by a cyclic group of symmetries. However, in both cases, this group of symmetries did not preserve the $\CN=1$ supercurrent of $V^{fE_8}$ or $V^{f\natural}$. As a consequence, when $F_{24}$ is constructed in this way, there is no $\CN=1$ superconformal structure inherited from the original SVOA. This raises the question whether $F_{24}$ with a given choice of $\CN=1$ superconformal structure can be obtained from $V^{fE_8}$ or $V^{f\natural}$ by an orbifold procedure, where the group of symmetries we quotient by preserves the superconformal current of $V^{fE_8}$ or $V^{f\natural}$, and the $\CN=1$ structure on $F_{24}$ is exactly the one induced by the parent theory. In this section, we will show that  all choices of $\CN=1$ structure on $F_{24}$ can be obtained from the $V^{fE_8}$ SVOA.

This result is interesting in view of the correspondence between chiral vertex operator superalgebras and non-chiral $\CN=(4, 4)$ supersymmetric nonlinear sigma models proposed in \cite{ACDV,derived,Cheng:2015kha,TW,CDR}. In particular, in \cite{ACDV} it was shown that there is a certain set of supersymmetry preserving automorphisms of $V^{fE_8}$ that are closely related to symmetries of supersymmetric sigma models on $T^4$, and that the orbifold of $V^{fE_8}$ by any such automorphism is either the SVOA $V^{f\natural}$ or $V^{fE_8}$. Similar relationships between the chiral SVOA $V^{s\natural}$ with $c=12$ (essentially $V^{f\natural}$) and the non-chiral $\CN=(4,4)$ K3 sigma models  with $c=\bar{c}=6$ have also been explored in previous works \cite{derived,Cheng:2015kha}; in \cite{TW,CDR}, it was shown that they can be related via a certain reflection procedure. Here, we show that $V^{fE_8}$ admits some further $\CN=1$ preserving automorphisms for which the orbifold is $F_{24}$. It would be very interesting to understand what the meaning of this result is on the sigma model side of the correspondence.

\subsection{Generalities}
Let us show that each choice of $\CN=1$ structure on $F_{24}$ can be obtained from an orbifold of the $E_8$ SVOA $V^{fE_8}$, with its standard $\CN=1$ structure, by a symmetry that commutes with the $\CN=1$ supercurrent.  Let $\psi^i$, $\partial X^i$, $i=1,\ldots,8$, and $V_\lambda$, $\lambda\in E_8$ be the fields of weights $1/2$, $1$, and $\lambda^2/2$, respectively, generating $V^{fE_8}$. The standard $\CN=1$ supercurrent is $G\sim :\sum_i \psi^i\partial X^i:$ (up to normalization). 

Let us first consider the case of $F_{24}$ with the $\CN=1$ structure corresponding to $g=A_1^8$.  We use a description of the $E_8$ lattice as $E_8=D_8 \cup (\chi+D_8)$, where
\be D_8=\{(x_1,\ldots,x_8)\in \ZZ^8\mid \sum_{i=1}^8 x_i\in 2\ZZ\}\ ,
\ee and $\chi=(1/2,1/2,\ldots,1/2)\in \RR^8$, so that
\be \chi+D_8=\{(x_1,\ldots,x_8)\in (\frac{1}{2}+\ZZ)^8\mid \sum_{i=1}^8 x_i\in 2\ZZ\}\ .
\ee We note that $\chi+D_8$ is one of the four cosets in $D_8^*/D_8$. Consider a symmetry $\delta$ of $V^{fE_8}$ that acts trivially on $\psi^i$, $\partial X^i$, $i=1,\ldots, 8$, and $V_{\lambda}$, for all $\lambda\in D_8$ but acts by $V_{\lambda}\mapsto -V_{\lambda}$ for $\lambda\in \chi+D_8$. This is a symmetry of order $2$ of the SVOA acting trivially on the supercurrent $G$. The group $\langle \delta\rangle\cong \ZZ_2$ is a subgroup of a $U(1)^8$ group of symmetries which preserves $\partial X^i$, $\psi^i$ (and therefore preserves the supercurrent) and acts by $V_\lambda\mapsto e^{2\pi i \alpha\cdot \lambda}V_\lambda$ for some $\alpha\in (E_8\otimes \RR)/E_8$. In particular, $\delta$ corresponds to $\alpha=(1,0,0,\ldots, 0)\in (E_8\otimes \RR)/E_8$. 

The orbifold of $V^{fE_8}$ by $\langle \delta\rangle$ is again an $\CN=1$ SVOA, with the supercurrent $G$ inherited from the parent theory. The $\delta$-invariant subalgebra $(V^{fE_8})^\delta$ is a supersymmetric lattice SVOA based on the lattice $D_8$. The $\delta$-invariant $\delta$-twisted sector is the $(V^{fE_8})^\delta$-module corresponding to the coset $\alpha+D_8$, another of the four cosets of $D_8^*/D_8$. Since $D_8\cup (\alpha+D_8)\cong \ZZ_8$, it turns out that the orbifold can be described as a lattice SVOA based on the odd unimodular lattice $\ZZ_8$, together with the $8$ fermions $\psi^i$. To check that this is actually the same as the SVOA generated by $24$ free fermions, notice that there are $24$ fields of weight $1/2$, namely $\psi^i$ and $V_{\pm e_i}$, $i=1,\ldots, 8$, where $\{e_i\}$, with $e_i=(0,\ldots,0,1,0,\ldots,0)\in \ZZ^8$, is the standard basis of $\ZZ_8$. Furthermore, since $:V_{e_i}V_{-e_i}:\sim \partial X^i$, these $24$ fields of weight $1/2$ generate the whole SVOA. Set $\lambda^i:=\psi^i$,  $\lambda^{8+i}:=V_{e_i}$, $\lambda^{16+i}:=V_{-e_i}$, $i=1,\ldots, 8$. Then, the supercurrent $G$ can be written (up to normalization) as
\be G\sim \sum_{i=1}^8 :\psi^i\partial X^i:\sim \sum_{i=1}^8 :\psi^iV_{e_i}V_{-e_i}:=\sum_{i=1}^8 :\lambda^i\lambda^{8+i}\lambda^{16+i}:\ ,
\ee which is of the form \eqref{freefermG} with $c_{ijk}$ the structure constants of $A_1^8$.

There are similar orbifolds of $V^{fE_8}$ giving all the other $\CN=1$ structures on $F_{24}$. In order to describe them, it is easier to implement the procedure in reverse, i.e. to find a cyclic group of symmetries of $F_{24}$ which preserves a given $\CN=1$ supercurrent and such that the orbifold theory is isomorphic to $V^{fE_8}$ with its $\CN=1$ structure. Then, one uses the fact that orbifolds by cyclic groups are `invertible'. This means that if a CFT $B$ is obtained from the CFT $A$ via an orbifold by a cyclic group $\langle \delta \rangle$, then the CFT $B$ has a `quantum symmetry' $Q$ such that the orbifold of $B$ by $\langle Q\rangle$ is again $A$. The symmetry $Q$ has the same order $N$ as $\delta$ and acts on $B$ by multiplying the states in the $\delta^r$-twisted sector by $e^{\frac{2\pi i r}{N}}$. In particular, if $A$ has a $\delta$-invariant $\CN=1$ supercurrent, the induced supercurrent in $B$ is also $Q$-invariant, because it resides in the untwisted sector. By applying this general procedure to the case we are interested in, then if we can show that $V^{fE_8}$ can be obtained from $F_{24}$ through an orbifold by an $\CN=1$-preserving cyclic group, we know that the orbifold of $V^{fE_8}$ by the `quantum symmetry' will give back $F_{24}$. 

To implement this construction, we need a symmetry $\sigma$ of $F_{24}$ that projects out most of the $24$ spin $1/2$ fields, leaving at most $8$ of them---this is the number of spin $1/2$ fields in $V^{fE_8}$. Furthermore, the currents that are supersymmetric descendants of these $\sigma$-invariant fermions must commute with each other---this is what happens with the supersymmetric descendants of the $8$ free fermions in $V^{fE_8}$.   In order to preserve the supercurrent $G$, it is sufficient that $\sigma$ acts on the $24$ fermions $\lambda^i$---whose supersymmetric descendants (currents) generate one of the Lie algebras $g$ listed in \eqref{N1algebras}---by an automorphism of the corresponding Lie algebra $g$. Explicitly, let $J^i$ be the current superpartner of the free fermion $\lambda^i$, $i=1,\ldots, 24$, and let $\theta$ be a Lie algebra automorphism acting as $J^i\to \theta(J^i)=\sum_j\theta_{ij}J^j$ on the currents. Then, we let the symmetry $\sigma$ act by $\lambda_i\mapsto \sigma(\lambda_i)\equiv \sum_j\theta_{ij}\lambda^j$. Since $\theta$ is an automorphism of $g$, it must preserve the structure constants $c_{ijk}$, which implies that $G\sim \sum_{i,j,k} c_{ijk}\lambda^i\lambda^j\lambda^k$ is also preserved by this symmetry. The condition that the superpartners of $\sigma$-invariant fermions must commute (i.e., they must be contained in some  Cartan subalgebra of the Lie algebra $g$) automatically ensures that there are at most $8$ spin $1/2$ fields surviving the orbifold projection, because the algebras listed in \eqref{N1algebras} have rank at most $8$. This condition can be achieved by taking $\theta$ to be an inner automorphism of $g$ in a given Cartan torus acting non-trivially on all non-zero roots. 

 A symmetry $\sigma$ projecting out all spin $1/2$ fields corresponding to non-trivial roots of $g$, and such that the orbifold is consistent, can be constructed as follows. Let $g=\oplus_k g_k$ be the decomposition of $g$ into simple components $g_k$, with Weyl vectors $\rho_k$ and dual Coxeter numbers $\dcox_{g_k}$. Then, the Weyl vector of $g$ is $\rho=\oplus_k \rho_k\in P_g=\oplus_k P_{g_k}$. With our normalization of the Killing form, the Freudenthal-de Vries strange formula reads
 \be (\rho_k|\rho_k)=\frac{\dim g_k}{12}\ ,
 \ee so that
 \be\label{Weylnorm} (\rho|\rho)=\sum_k (\rho_k|\rho_k)=\frac{\dim g}{12}=2\ .
 \ee

We take the symmetry $\sigma$ to act on the fermion $\lambda^\alpha$ corresponding to a root $\alpha$ by $\sigma(\lambda^\alpha)=e^{2\pi i(\rho|\alpha)}\lambda^\alpha$, and to act trivially on the spin $1/2$ fields corresponding to the Cartan subalgebra. Notice that, for each positive root $\alpha$ in the $g_k$ component, we have
\be  0<(\rho|\alpha) \le (\rho|\theta_k)\ ,
\ee where $\theta_k$ is the highest root of $g_k$. With our normalization for the Killing form, we have
\be (\rho|\theta_k)=(\rho_k|\theta_k)=1-\frac{1}{\dcox_{g_k}}\ ,
\ee so that $0<(\rho|\alpha)<1$ for all positive roots. In particular, $\sigma$ acts non-trivially on all $\lambda^{\pm \alpha}$, so that only the spin $1/2$ fields corresponding to the Cartan subalgebra are preserved by the orbifold projection.

In order to check that the orbifold is consistent, one needs to check the level-matching condition, i.e. to  verify that the levels of the $\sigma$-twisted NS sector are valued in $\frac{1}{2N}\ZZ$, where $N$ is the order of $\sigma$. In general, for a theory of $2n$ free fermions with a symmetry $\sigma$ of order $N$ acting with eigenvalues $e^{\pm 2\pi i r_i}$, $i=1\ldots,n$, $r_i\in \frac{1}{N}\ZZ$,  the $\sigma$-twisted NS states have conformal weights valued in  $\frac{1}{2 }\sum_i r_i^2+\frac{1}{2N}\ZZ$.
In particular, if we take $|r_i|\le \frac12$, then the $\sigma$-twisted ground states have conformal weight exactly $\frac{1}{2 }\sum_i r_i^2$. This standard formula can be obtained, for example, by writing the Virasoro generators $L_n$ in the twisted sector in terms of normal ordered products of fermionic generators, and fixing the normal ordering constant in $L_0$ by requiring that the relation $[L_1,L_{-1}]=2L_0$ is satisfied (see, for example, \cite{Polchinski1}). Applying this formula to our case, we obtain
\be \frac{1}{2 }\sum_i r_i^2=\frac{1}{2}\sum_{\alpha\in\Delta^+} (\rho|\alpha)^2=\frac{1}{2}(\rho|\rho)= 1 \in \frac{1}{2N}\ZZ\ .
\ee Thus, the conformal weights are valued in $\frac{1}{2N}\ZZ$, and the level matching condition is satisfied. We stress that it is not necessarily true that $|(\rho|\alpha)|\le 1/2$ for all $\alpha\in\Delta^+$, so this formula does not imply that the conformal weights of the ground states are always $1$.

We conclude that the orbifold of $F_{24}$ by $\sigma$ is a consistent holomorphic SVOA of central charge $12$, so the only possibilities are $V^{fE_8}$, $F_{24}$ or $V^{f\natural}$. The latter case can be easily ruled out: the orbifold theory contains at least the spin $1/2$ fields $\lambda^i$ corresponding to the Cartan subalgebra of $g$, while $V^{f\natural}$ contains no such fields. Finally, we verified in a case-by-case analysis that the orbifold theory never contains $24$ fields of spin $1/2$, so we conclude that the orbifold is the $V^{fE_8}$ theory.

For example, for an $A_n$ algebra, the automorphism $\sigma$ acts by multiplication by $e^{\frac{2\pi i}{n+1}}$ on the root  space $g_{\alpha_i}$ for every simple root $\alpha_1,\ldots, \alpha_n$; then, for any root $\alpha=\sum_i m_i \alpha_i$ the automorphism acts on the root spaces $g_{\alpha}$ by multiplication by $e^{\frac{2\pi i\sum_i m_i}{n+1}}$; since $\sum_i m_i\le n$ (and $\sum_i m_i \ge -n$ for negative roots) for the $A_n$ algebra \footnote{This is automatic by one of the definitions of the Coxeter number.}, one has that  $\sigma$ acts trivially only on the Cartan subalgebra. Besides the currents in the Cartan subalgebra of $g$, which are supersymmetric descendants, the symmetry $\sigma$ leaves invariant a number of superconformal primary currents. These can be easily determined since we know the eigenvalues of $\sigma$ on $\lambda^i$. Below we shall consider some explicit examples.

\subsection{Examples}
In this section, we summarize the action of $\mathcal N=1$--preserving orbifolds of $F_{24}$ which reproduce the SVOA $V^{fE_8}$, for each choice of $g$ in $F_{24}$. Let $\sigma_g$ be the orbifold symmetry which relates $F_{24}$ with $\mathcal N=1$ supercurrent determined by $g$ to $V^{fE_8}$, and $\lambda_g$ be the set of 24 eigenvalues of $\sigma_g$. As explained in the previous section, these eigenvalues can be computed independently for each simple component $g_k$ of each choice of semisimple Lie algebra $g= \oplus_k g_k$, such that $\lambda_g= \cup_k \lambda_{g_k}$. See Table \ref{tbl:sigmas} for a summary of these eigenvalues for each choice of simple Lie algebra which arises in our construction.

\begin{table}[htb]
\begin{center}
\begin{tabular}{c|c|c}
Simple Lie algebra $g_k$ &Dimension& Eigenvalues $\lambda_{g_k}$\\\hline\hline
$A_1$ & 3& $\{1,-1,-1\}$\\\hline
$A_2$ &8 &$ \{ 1,1, \ex({1\over 3}),\ex({1\over 3}),\ex({1\over 3}),\ex({2\over 3}),\ex({2\over 3}),\ex({2\over 3}) \}$\\\hline
$A_3$ &15&\makecell{$ \{1,1,1,-1,-1,-1,-1, \ex({1\over 4}),\ex({1\over 4}),\ex({1\over 4}),$\\$\ex({1\over 4}),\ex({3\over 4}),\ex({3\over 4}),\ex({3\over 4}),\ex({3\over 4})\}$}\\\hline
$A_4$ &24&\makecell{ $\{1,1,1,1,\ex({1\over 5}),\ex({1\over 5}),\ex({1\over 5}),\ex({1\over 5}),\ex({2\over 5}),\ex({2\over 5}),\ex({2\over 5}),\ex({2\over 5})$,\\$\ex({3\over 5}),\ex({3\over 5}),\ex({3\over 5}),\ex({3\over 5}),\ex({4\over 5}),\ex({4\over 5}),\ex({4\over 5}),\ex({4\over 5}) \}$}\\\hline
$B_2$ &10&$\{1,1,-1,-1,\ex({1\over 3}),\ex({1\over 3}),\ex({2\over 3}),\ex({2\over 3}),\ex({1\over 6}),\ex({5\over 6}) \}$\\\hline
$B_3$ &21&\makecell{ $\{1,1,1,-1,-1,\ex({1\over 5}),\ex({1\over 5}),\ex({1\over 5}),\ex({2\over 5}),\ex({2\over 5}),\ex({2\over 5}),$\\$\ex({3\over 5}),\ex({3\over 5}),\ex({3\over 5}),\ex({4\over 5}),\ex({4\over 5}),\ex({4\over 5}),\ex({1\over 10}),\ex({3\over 10}),\ex({7\over 10}),\ex({9\over 10}) \}$}\\\hline
$C_3$ &21&\makecell{ $\{1,1,1,-1,-1,-1,-1,\ex({1\over 4}),\ex({1\over 4}),\ex({1\over 4}),\ex({3\over 4}),\ex({3\over 4}),\ex({3\over 4}),$\\$\ex({1\over 8}),\ex({1\over 8}),\ex({3\over 8}),\ex({3\over 8}),\ex({5\over 8}),\ex({5\over 8})\ex({7\over 8}),\ex({7\over 8}) \}$}\\\hline
$G_2$ &14&\makecell{$ \{1,1,-1,-1,\ex({1\over 4}),\ex({1\over 4}),\ex({3\over 4}),\ex({3\over 4}),\ex({1\over 3}),\ex({2\over 3}),$\\$\ex({1\over 12}),\ex({5\over 12}),\ex({7\over 12}),\ex({11\over 12})\}$}
\end{tabular}\caption{Eigenvalues  for symmetries $\sigma_g$ which relate $F_{24}$ with $g$--preserving $\mathcal N=1$ supercurrent to $V^{fE_8}$. The 24 eigenvalues $\lambda_g$ of $\sigma_g$ can be decomposed into sets of eigenvalues which act on each simple component of $g$, such that $\lambda_g= \cup_k \lambda_{g_k}$ whenever $\sigma_g= \oplus_k \sigma_{g_k}$. Here we use the abbreviation $\ex(x) := e^{2\pi i x}$.}\label{tbl:sigmas}
\end{center}
\end{table}
Below we give more details of these orbifolds for several choices of Lie algebra $g$:
\begin{itemize}
\item {$g=A_1^8$} 

We have already studied this case, albeit beginning with an orbifold of $V^{fE_8}$; let us consider the same algebra from an orbifold of $F_{24}$. We take the automorphism $\sigma$ to act by $+1$ on the $8$ spin $1/2$ fermions corresponding to a Cartan subalgebra of $A_1^8$, which we denote by $\lambda^{a+}$, $a=1,\ldots,8$, and to act by $-1$ on the $16$ spin $1/2$ fermions corresponding to non-zero roots, which we denote by $\lambda^{a-}$, $a=1,\ldots,16$. The untwisted sector contains the $8$ fields $\lambda^{a+}$ of spin $1/2$, as well as the currents $\lambda^{a+}\lambda^{b+}$, $1\le a<b\le  8$, and $\lambda^{a-}\lambda^{b-}$, $1\le a<b\le  16$, which together form a $so(8)\oplus so(16)$ algebra of dimension $28+120=148$. In the $\sigma$-twisted (NS) sector, the fermions $\lambda^{a+}$ have a mode expansion $\lambda^{a+}=\sum_r \lambda^{a+}_rz^{-r-1/2}$ with $r\in \frac{1}{2}+\ZZ$, while the $\lambda^{a-}$ have modes $\lambda^{a-}_r$ with $r\in \ZZ$. The ground state level can be easily computed to be $16*\frac{1}{4}(1/2)^2=1$. The ground states must form a representation of the Clifford algebra of the $16$ zero modes $\lambda^{a-}_0$, so they must be degenerate with multiplicity $2^8=256$. Since each of these zero modes changes the $\sigma$-eigenvalue of the states, half of these $256$ ground states are $\sigma$-invariant and half have $\sigma$-eigenvalue $-1$. It follows that, after the orbifold projection, there will be $128$ additional currents, commuting with the $so(8)$ generated by $\lambda^{a+}\lambda^{b+}$ and transforming in a spinor representation of the $so(16)$ generated by $\lambda^{a-}\lambda^{b-}$. Together, the untwisted $so(16)$ currents and the $128$ $\sigma$-twisted currents form a copy of the $E_8$ current algebra.

\item {$g= A_2^3$} 

We let $\sigma$ act trivially on the six $\lambda^i$ corresponding to the Cartan subalgebra and by multiplication by $\omega=e^{\frac{2\pi i}{3}}$ on each of the spin $1/2$ fields corresponding to simple roots. Therefore, the eigenvalues of $\sigma$ on the $24$-dimensional representation of spin $1/2$ fermions are: $1$ with multiplicity $6$,  $e^{\frac{2\pi i}{3}}$ with multiplicity $9$, and $e^{-\frac{2\pi i}{3}}$ with multiplicity $9$. The untwisted sector currents are $\lambda_a^1\lambda_b^1$, $1\le a<b\le 6$, and $\lambda^\omega_a\lambda^{\bar\omega}_b$, $1\le a,b,\le 9$ (note that we do not require $a<b$ here). We have a total of $15+81=96$ currents, forming an $so(6)\oplus u(9)\cong so(6)\oplus u(1)\oplus su(9)$ algebra. In the $\sigma$-twisted sector, the $\lambda_{a}^{\omega}$ are moded in $\frac{1}{6}+\ZZ$ and the $\lambda_{a}^{\bar\omega}$ are moded in $\frac{5}{6}+\ZZ$ (and vice versa in the $\sigma^2$-twisted sector). Thus, the level of the twisted sector ground state is $18\times \frac{1}{4}(1/3)^2=1/2$. There are no zero modes, so that there is a unique ground state in each of the twisted sectors, and we can choose them to be $\sigma$-invariant. This gives $6$ untwisted spin-$1/2$ fields and one more from each of the two twisted sectors, for a total of $8$ spin-$1/2$ fields. From each twisted sector, we have six additional $\sigma$-invariant currents of the form $\lambda_{a,-1/2}^{1}|gr\rangle$, $a=1,\ldots,6$, and $9\cdot 8\cdot 7/3!=84$ of the form $\lambda^{\omega}_{a,-1/6}\lambda^{\omega}_{b,-1/6}\lambda^{\omega}_{c,-1/6}|gr\rangle$, $1\le a<b<c\le 9$  obtained by acting on the ground state $|gr \rangle$. The untwisted $so(6)\oplus u(1)$, together with the six currents $\lambda_{a,-1/2}^1|gr\rangle$ from each of the two twisted sectors, combine into an $so(8)$ algebra. The untwisted $su(9)$, together with the $84$ $\lambda^{\omega}_{a,-1/6}\lambda^{\omega}_{b,-1/6}\lambda^{\omega}_{c,-1/6}|gr\rangle$ from each twisted sector combine to form the $E_8$ algebra.

\item {$g=A_3A_1^3$} 

We let $\sigma$ act as it did for $g=A_1^8$ on the subset of the $\lambda_i$ corresponding to $A_1^3$, and let $\sigma$ act by multiplication by $i$ on the simple roots of $A_3$. The eigenvalue distribution in the form (multiplicity $\times$ eigenvalue) in the $24$-dimensional representation of the free fermions is $(6 \times 1)$, $(10\times -1)$, $(4\times i)$, $(4\times -i)$. The invariant currents form the algebra $so(6)\oplus so(10)\oplus u(4)\cong so(6)\oplus so(10)\oplus u(1)\oplus su(4)$, for a total of $15+45+16=76$ untwisted sector currents. The ground states of the $\sigma$- and $\sigma^3$-twisted sector have level $\frac{1}{4}(4*(1/4)^2+4*(1/4)^2+10*(1/2)^2)=3/4$; there are $2^5=32$ degenerate ground states, forming a representation of the Clifford algebra of the zero modes of the $10$ free fermions with eigenvalues $-1$. Half of them have a $\sigma$-eigenvalue $\zeta$ while the other half have eigenvalue $-\zeta$. We fix the action of $\sigma$ on the twisted sectors in such a way that $\zeta=i$. This means that in each of the $\sigma$- and $\sigma^3$-twisted sectors, there are $4\cdot 16=64$ $\sigma$-invariant currents, transforming in a $(16,4)$-representation of $so(10)\oplus u(4)$.  Since $\sigma^2$ has an eigenvalue distribution $(16 \times 1)$, $(8\times -1)$, the $\sigma^2$-twisted ground states have level $8\times \frac{1}{4}*(1/2)^2=\frac{1}{2}$, and form a $2^4=16$ dimensional representation of the Clifford algebra of the $8$ fermionic zero modes. The $\sigma$-eigenvalue distribution on the ground states is $(2\times 1)$, $(6\times -1)$, $(4\times i)$, $(4\times -i)$. Currents in the $\sigma^2$-twisted sector are obtained by acting on the ground states with one of the $16$ half-integrally moded fermions, which have $\sigma$-eigenvalues $+1$ ($6$ of them) or $-1$ ($10$ of them). Therefore, we get $6\times 2+10\times 6=72$ $\sigma$-invariant currents. In total, we have $76+64+72+64=276$ currents, as expected. In particular, the $so(6)\oplus u(1)$ algebra, together with the $6\times 2$ currents in the $\sigma^2$-twisted sector, form the `fermionic' $so(8)$  obtained from the OPE of two spin $1/2$ fields, while the $so(10)\oplus su(4)$ algebra in the untwisted sector combines with the $64$ in each of the $\sigma$- and $\sigma^3$-twisted sectors (in the $(16,4)$ and $(16,\bar 4)$ representation of $so(10)\oplus su(4)$) and the $10\times 6$ from the $\sigma^2$-twisted sector (in the $(10,6)$ representation of $so(10)\oplus su(4)$) to form the $E_8$ algebra.

\item {$g=G_2B_2$ \& the rest} 

For the other algebras, in particular with non-simply-laced components, the analysis is slightly more complicated. We illustrate the general procedure by describing one example, the $g=G_2B_2$ case. The symmetry $\sigma$ has order $12$ and fixes $4$ fermions, corresponding to the Cartan subalgebra of $g$. On the remaining $20$ fermions, the eigenvalues are $e^{\pm 2\pi r_i}$, where the $r_i\in \frac{1}{12}\ZZ$, $i=1,\ldots, 10$, are $\frac{1}{4}$, $\frac{1}{12}$, $\frac{1}{3}$, $\frac{5}{12}$, $\frac{1}{2}$,  $\frac{3}{4}$ from the $G_2$ component, and $\frac{1}{3}$, $\frac{1}{6}$,  $\frac{1}{2}$, $\frac{2}{3}$ from the $B_2$ component. 
From these data, one can compute the conformal weights of the $\sigma^n$-twisted ground states (NS sector), obtaining $\frac{7}{12}$ for $n=1,5,7,11$,  $\frac{1}{2}$ for $n=2,10$,  $\frac{1}{2}$ for $n=3,9$, $\frac{1}{3}$ for $n=4,8$, $\frac{1}{2}$ for $n=6$.
 As expected, the level matching condition is satisfied. The degeneracy of the ground states is determined by the number of fermionic zero modes in each twisted sector. Using the algebra of fermionic oscillators, one can determine the number of spin-$1/2$ states in each twisted sector. Next, one needs to project on the $\sigma$-invariant subspace. On each twisted sector, the action of $\sigma$ can be easily determined up to an overall phase. The general theory of orbifolds tells us that, when the level matching condition is satisfied, there exists a choice for these phases such that the $\sigma$-invariant fields define a consistent SVOA. However, determining the right phases explicitly is usually complicated, so in practice we take the following shortcut for this example. We verified that, for any choice of the phases, the total number of $\sigma$-invariant spin $1/2$ fields is less than $24$ (and more than $0$). This is sufficient to conclude that the orbifold theory is $V^{fE_8}$. As a consistency check, we also verified that there is a choice of phases for which the number of $\sigma$-invariant spin $1/2$ fields  is exactly $8$. We performed a similar analysis for all choices of $g$, and in all cases we obtained $V^{fE_8}$ as the orbifold theory.
 \end{itemize}
 
 With the analysis of $F_{24}$ orbifolds complete, we will now turn to a uniform construction of BKM superalgebras for each choice of $g$. 

\section{BKM superalgebras from $\CN=1$ SVOAs: a general construction}\label{s:BKM}

In this section, we describe a procedure to construct a BKM superalgebra starting from a holomorphic SVOA $V$ with central charge $c=12$ with an $\CN=1$ structure. The procedure is heavily inspired by the definition of physical states in superstring theory---it is, in a sense, a `chiral version' of that construction---and is a supersymmetric generalization of Borcherds' construction of the fake Monster and Monster Lie algebras \cite{BorcherdsFake,BorcherdsMM} that was inspired by bosonic string theory.  The main steps have been developed in \cite{Sch1} for the specific example where $V$ is the lattice SVOA $V^{fE_8}$ based on the $E_8$ lattice, and has been generalized to the example where $V=V^{f\natural}$ in \cite{Harrison:2018joy} . We will  briefly describe the main steps in this construction and refer to \cite{Harrison:2018joy,Sch1,LZ} for the proofs of most statements.

\subsection{Super vertex algebras}
For this construction, we need to consider a super vertex algebra (SVA) $V^{tot}$ given by a product
\be V^{tot}=V^m\otimes V^{gh}=V^{int}\otimes V^{X,\psi}\otimes V^{gh}
\ee where the `matter' SVA $V^m$ has central charge $15$ and the `ghost' SVA $V^{gh}$ has central charge $-15$. The matter SVA $V^m$ is itself a product of an `internal' SVOA $V^{int}$ of central charge $12$ and a `space-time' SVA $V^{X,\psi}$ of central charge $3$.

Each of these SVAs has a $\ZZ_2$-grading given by a fermion parity $(-1)^F$, and there is a canonically twisted module on which one can choose an action of the fermion parity $(-1)^F$ (though with a certain ambiguity). We again use the physics parlance: the vertex superalgebras are the Neveu-Schwarz (NS) sector and the twisted module is the Ramond (R) sector. For this reason, we will often put a subscript $NS$ on the SVAs, such as $V^{int}_{NS}\equiv V^{int}, V^{gh}_{NS}\equiv V^{gh}$ etc. 

In more detail, the various factors are as follows:
\begin{itemize}
	\item  The `internal (Neveu-Schwarz) sector' of the superstring theory $V^{int}_{NS}=V^{int}$ is a $\CN=1$ self-dual SVOA  (holomorphic SCFT) of central charge $12$. Up to a choice of the $\CN=1$ supercurrent, there are only three possible such SVOAs, up to isomorphism \cite{CDR}. One is given by the $24$ free fermion SVOA $F_{24}$ described in section \ref{s:F24}; the other two are the supersymmetric $E_8$ lattice SVOA $V^{fE_8}$ and the Conway module $V^{f\natural}$ studied in \cite{Duncan}. The name `internal' comes from the idea of compactifying the $10$-dimensional spacetime of a type II superstring on an $8$-dimensional compact manifold, whose corresponding non-linear sigma model is a SCFT of central charge $12$. Note however that the standard superstring construction has also an anti-holomorphic sector, while our construction in this article is chiral. The space $V^{int}_{NS}(\frac{1}{2})$ of states of conformal weight ($L_0$-eigenvalue) $1/2$ is $24$-dimensional for $F_{24}$, $8$ dimensional for $V^{fE_8}$ and $0$-dimensional for $V^{f\natural}$. The space $G^V_{-1/2}V^{int}_{NS}(\frac{1}{2})\subseteq V^{int}_{NS}(1)$ of their superpartners will be relevant in the following. The zero modes of these currents generate a finite-dimensional Lie algebra $g$, where $g=0$ for $V^{f\natural}$ and $g=u(1)^{\oplus 8}$ for $V^{fE_8}$; for $F_{24}$ the possible algebras $g$, which are non-abelian and depend on the choice of the supercurrent $G^V$, are described in section \ref{s:F24}.   The canonically twisted module (Ramond sector) is denoted by $V^{int}_R$.
	\item The `uncompactified' directions are represented by an SVA $V^{X,\psi}\equiv V^{X,\psi}_{NS}$ (the NS sector) based on the even unimodular lattice $\Gamma^{1,1}$ of signature $(1,1)$, and its canonically twisted module $V_R^{X,\psi}$.  The basic fields are two chiral free bosons $ X^+(z),X^-(z)$
	 and their superpartners, the free fermions $\psi^+,\psi^-$.
	The chiral bosons alway appear either with derivatives $\partial^n X^\pm$, $n\ge 1$, or exponentiated in the form of vertex operators $e^{ikX}$ for each $k\in \Gamma^{1,1}$. A convenient description of $\Gamma^{1,1}$  is as the lattice of vectors $k\equiv (k^+,k^-)=(m,n)\in \ZZ\oplus \ZZ$ with quadratic  form
	\be k^2\equiv k^\mu k_\mu\equiv \eta_{\mu\nu} k^\mu k^\nu=-2k^+k^-=-2mn\ .
	\ee Roughly speaking, $X^+,X^-$ represent the light-cone coordinates in a $1+1$ dimensional space-time $\RR^{1,1}$ with metric $\eta_{++}=\eta_{--}=0$, $\eta_{+-}=-1$.  The mode expansions are
	\be i\partial X^\mu=\sum_{n\in \ZZ} \alpha^\mu_n z^{-n-1}\ ,\qquad \mu\in \{+,-\}\ .
	\ee
	\be \psi^\mu(z)=\sum_{r\in \ZZ+\nu} \psi^\mu_r z^{-r-1/2}\ ,\qquad \mu\in \{+,-\}\ .
	\ee where $\nu=1/2$ in the NS sector and $\nu=0$ in the R sector. The vertex operators $e^{ikX}$ correspond to states $|k\rangle$ that are eigenstates of the the zero modes $P^\mu=\alpha_0^\mu$ (momentum operators) of $i\partial X^\mu$ with eigenvalues $k^\mu$, $\mu\in \{+,-\}$. 
	The stress energy tensor $T^{X,\psi}(z)$ and $\CN=1$ supercurrent $G^{X,\psi}(z)$ are given by
	\be T^{X,\psi}(z)=\frac{1}{2}:\partial X^\mu \partial X_\mu:(z) +\frac{1}{4} :\psi_\mu \partial \psi^\mu:(z)
	\ee
	\be G^{X,\psi}(z)=:\psi^\mu\partial X_\mu:(z)
	\ee and generate an $\CN=1$ superalgebra with central charge $c=3$. The fields $\psi^\mu$ and $e^{ikX}$ are superconformal primaries with conformal weight $1/2$ and $k^2/2$, respectively; the fields $\partial X^\mu$ are superconformal descendants of $\psi^\mu$ and have weight $1$.\\
	We refer to the product of the SVAs $V^{int}$ and $V^{X,\psi}$ as the \emph{matter sector} $V^m=V^{int}\otimes V^{X,\psi}$, with matter stress-energy tensor
	\be T^m(z)=T^V(z)+T^{X,\psi}(z),
	\ee
	$\CN=1$ supercurrent
	\be G^m(z)=G^V(z)+G^{X,\psi}(z),
	\ee and total central charge $c=15$.
	\item The ghost sector $V^{gh}$ is a SVA generated by the anticommuting bosonic fields
	\be b(z)=\sum_{n\in \ZZ} b_n z^{-n-2}\qquad c(z)=\sum_{n\in \ZZ} c_n z^{-n+1}\ ,
	\ee and their superpartners, the commuting fermionic fields 
	\be \beta (z)=\sum_{r\in \ZZ+\nu } \beta_r z^{-r-3/2}\qquad \gamma(z)=\sum_{r\in \ZZ + \nu} \gamma_r z^{-r+1/2}\ ,
	\ee where $\nu =1/2$ in the NS sector $V^{gh}_{NS}\equiv V^{gh}$ and $\nu=0$ in the Ramond sector $V^{gh}_{R}$ (the canonically twisted module of $V^{gh}$). The stress-energy tensor $T^{gh}(z)$ and the $\CN=1$ supercurrent $G^{gh}(z)$ are given by
	\be T^{gh}= -:(\partial b) c:-2 :b\partial c:-\frac{1}{2}:(\partial\beta)\gamma:-\frac{3}{2}:\beta\partial\gamma:
	\ee 
	\be G^{gh}(z)=-(\partial\beta) c(z)-\frac{3}{2}\beta\partial c(z)  -2b\gamma(z)
	\ee and form an $\CN=1$ superVirasoro algebra with central charge $c=-15$. The fields $b$ and $c$ and their superpartners $\beta$ and $\gamma$ have conformal weights $2$, $-1$,  $3/2$ and $-1/2$, respectively. One can define a ghost quantum number, with respect to which $c$ and $\gamma$ has charge $+1$, while $b$ and $\beta$ have charge $-1$. It is often useful to have an alternative description of the superghosts $\beta,\gamma$ as a subalgebra of a SVA generated by two anticommuting fields $\xi$ and $\eta$ of conformal weight $0$ and $1$  and a chiral scalar $\phi$
	\be \beta=\partial \xi e^{-\phi}\qquad \gamma=\eta e^{\phi}\ .
	\ee The fields $\eta, \xi$ obey the same OPE as $b$ and $c$, while $\phi$ generates a lattice vertex algebra based on a $1$-dimensional lattice, and always appears with derivatives $\partial^n\phi$, $n>0$ or in exponentials $e^{\frac{m}{2}\phi}$, $m\in \ZZ$. More precisely, the fields $e^{\frac{m}{2}\phi}$, have $m$ even or odd depending on whether they act on the NS or the R sector. The stress energy tensor becomes \be T^{gh}= -:(\partial b) c:-2 :b\partial c:=-\frac{1}{2}\partial\phi\partial\phi-\partial^2\phi-\eta\partial\xi\ 
	\ee in terms of these fields. Note that the SVA generated by $\eta,\xi$ and $\phi$ is strictly larger than the one generated by $\beta$ and $\gamma$.
	
The $\beta\gamma$-module built starting from the $PSL(2,\CC)$-invariant vacuum $|0\rangle$ is unbounded from below, since the states $(\gamma_{1/2})^n|0\rangle$ have arbitrarily low $L_0$-eigenvalues. More generally, one can consider different $\beta\gamma$-modules, starting from a state $|p\rangle$, $p\in \ZZ$ (NS sector) or $p\in\frac{1}{2}+\ZZ$ (R sector), satisfying
\begin{align}\label{pvaccum1} 
&\beta_r|p\rangle=0, & r\ge -p-1/2\\
\label{pvaccum2}&\gamma_r|p\rangle=0, & r> p+1/2\ .
\end{align} 
$p$ is the picture number of the $\beta\gamma$-module. It is easy to see that only for $p\in\{-1,-1/2,-3/2\}$, all positive modes of both $\beta$ and $\gamma$ annihilate $|p\rangle$; therefore, only in this case the $L_0$ eigenvalues are bounded from below with  $|p\rangle$ having the lowest eigenvalue (ground state).  
The different $\beta\gamma$-modules are related to each other in the larger algebra generated by $\xi,\eta,\phi$, by
\be |p\rangle =e^{p\phi}|0\rangle.
\ee
\end{itemize}

The full SVA  \be V^{tot}_{NS}\equiv V^{tot}=V^{int}\otimes V^{X,\mu}\otimes V^{gh}\equiv V^m\otimes V^{gh}\ \ee  contains a $\CN=1$ superVirasoro subalgebra with central charge $c^{tot}=0$ generated by the stress-energy tensor  $T(z)=\sum_n L_n z^{-n-2}$
\be T=T^m+T^{gh}=T^V+T^{X,\psi}+T^{gh}
\ee and $\CN=1$ supercurrent
\be G=G^m+G^{gh}=G^V+G^{X,\psi}+G^{gh}\ .
\ee
The fermion number operator on $V^{tot}$, leaving bosons fixed and multiplying fermions by $-1$, is the product of the fermion number operators on the single factors. The canonically twisted module (Ramond sector) of $V^{tot}$ is just the product 
\be V^{tot}_R=V^{int}_R\otimes V^{X,\mu}_R\otimes V^{gh}_R=V^m_R\otimes V^{gh}_R \ ,\ee of the Ramond sectors, where we defined the matter Ramond sector $V^m_R=V^{int}_R\otimes V^{X,\mu}_R$. For each of these SVA,  the action of the fermion number operator on the algebra can be extended to an action on the Ramond sector. There is a certain ambiguity in choosing this expansion; we assume that a choice has been made, so that $(-1)^F$ has order $2$ on the Ramond sector.

\subsection{BRST cohomology}
The next step on the path to obtaining the chiral physical states, mimicking the usual superstring construction, is to perform a GSO projection, i.e. to consider only the even subspace \be
V^{GSO}\equiv V^{tot}_{NS+}\oplus V^{tot}_{R+}\ ,\ee
which is the eigenspace of the total fermion number with eigenvalue $+1$, and then to further restrict to the kernel of $b_0,L_0$ (NS sector) or to the kernel of $b_0,L_0,\beta_0,G_0$ (R sector)\footnote{Strictly speaking, it is not clear that one needs to impose the extra restrictions to $\ker \beta_0, \  \ker G_0$ in the R sector, though it is a choice that is sometimes made for convenience. The states at nonzero momentum are insensitive to this restriction, but computations of the cohomology of states at zero-momentum can become simpler. In \cite{HPPV} we will study a non-chiral superstring theory based on $V^{f\natural}\otimes \bar{V}^{f\natural}$ and compute the physical states without imposing this extra R sector condition.}
\be C=(V^{tot}_{NS+}\cap \ker \langle b_0,L_0\rangle)\oplus (V^{tot}_{R+}\cap \ker \langle b_0,L_0,\beta_0,G_0\rangle)\ ,
\ee where $\langle b_0,L_0\rangle$ and $\langle b_0,L_0,\beta_0,G_0\rangle$ denote the subalgebras generated by the corresponding elements.
On this space, we introduce some gradings given by the ghost and picture numbers $n\in \ZZ$ and $p\in \frac{1}{2}\ZZ$, and the momentum $k=(k^+,k^-)\in \Gamma^{1,1}\cong\ZZ\oplus \ZZ$
\be C=\oplus_{k \in \Gamma^{1,1}} C(k)=\oplus_{k,p,n} C^n_{p}(k)\ .
\ee Notice that the NS and the Ramond sector can be distinguished by their picture number $p$, which is integral in the NS sector and half-integral in the Ramond sector.  We now introduce the BRST charge
\be Q=\sum_n c_{-n}L^m_n +\sum_r \gamma_{-r}G^m_r+\frac{1}{2} (:cT^{gh}:)_0+\frac{1}{2}(:\gamma G^{gh}:)_0\ ,\ee which is a nilpotent operator that commutes with $p,k$ and shifts the ghost number by $1$. For each picture number $p$ and momentum $k\in \Gamma^{1,1}$, we have a complex
\be \ldots C^{n-1}_{p}(k) \stackrel{Q}{\longrightarrow} C^n_{p}(k) \stackrel{Q}{\longrightarrow} C^{n+1}_{p}(k) \stackrel{Q}{\longrightarrow} \ldots
\ee and we define the corresponding cohomology spaces as $H^n_{p}(k)$.

Let us recall some results about this cohomology:
\begin{description}
	\item[Picture changing] For each $n,k$, there is a `picture raising operator' homomorphism \be X:H^n_{p}(k)\to H^n_{p+1}(k)
	\ee which is an isomorphism if $k\neq 0$. Therefore, at least at nonzero momentum, one is led to consider only the cohomology groups in the `canonical pictures'  $p=-1$ (NS sector) and $p=-1/2, -3/2$ (R sector). It is also reasonable to expect that one can restrict to these canonical pictures at zero momentum without losing any interesting information, and we do so throughout this text. 
	\item[Canonical ghost number] For $k\neq 0$ and $n\neq 1$,
	\be H^n_{p}(k)=0\ .
	\ee This theorem is proved in section $3$ of \cite{LZ}; an alternative proof is given in \cite{FigueroaKimura}. The proof only uses the fact that for $k\neq 0$ the matter sector is a free module for the superalgebra generated by the negative modes of the matter superVirasoro algebra. This is true for any critical internal SVOA $V^{int}$, so the theorem generalizes immediately to the case we are considering. It fails for $k=0$ because the module is not free in that case: there are relations corresponding to the fact that $G_{-1/2}$ and $L_{-1}$ annihilate the $PSL(2,\CC)$-invariant vacuum state of the matter SVA $V^m$. 
	\item[Bilinear form] There is a non-degenerate  bilinear form $( \cdot, \cdot)_H$ pairing $H^{n}_p(k)$ with $H_{-2-p}^{2-n}(-k)$, which is defined in terms of the bilinear forms on the matter and ghost vertex algebras. In particular, $( \cdot, \cdot)_H$ restricts to a non-degenerate form on $\oplus_k H^{1}_{-1}(k)$. In the Ramond sector, this bilinear form non-degenerately pairs $H^{1}_{-1/2}(k)$ and $H^{1}_{-3/2}(-k)$. For $k\neq 0$, combining this bilinear form with the spectral flow isomorphism $X:H^{1}_{-3/2}(k)\to H^{1}_{-1/2}(k)$, we get a non-degenerate bilinear form on $\oplus_{k\neq 0}H^{1}_{-1/2}(k)$. For $k=0$, the homomorphism $X$ might have a non-trivial kernel and the induced bilinear form on $H^{1}_{-1/2}(0)$ might be degenerate. We will deal with the $k=0$ case separately in the following. 
	\item[Equivalence with light-cone quantization] For $k\neq 0$, the no-ghost theorem ensures that there is an isomorphism of vector spaces 
	$$H^1_{-1}(k)\cong V^{int}_{NS-}(-\frac{k^2}{2}+\frac{1}{2})\ ,\qquad \qquad k\neq 0\ ,$$
	where $V^{int}_{NS-}(h)$ denotes the component of the internal SVOA $V^{int}$ with $L_0$-eigenvalue $h$ and negative fermion number (the latter condition is automatically satisfied, since $h=-\frac{k^2}{2}+\frac{1}{2}\in \frac{1}{2}+\ZZ$).  Similarly, there is an isomorphism of vector spaces 
	$$H^1_{-1/2}(k)\cong V^{int}_{R+}(-\frac{k^2}{2}+\frac{1}{2})\cong V^{int}_{R-}(-\frac{k^2}{2}+\frac{1}{2})\ ,\qquad \qquad k^2\neq 0\ ,$$
	where $V^{int}_{R+}(h)$ and $V^{int}_{R-}(h)$ denote the components of the canonically twisted module $V^{int}_R$ of the internal SVOA $V$ with $L_0$-eigenvalue $h$ and with positive (respectively, negative) fermion number. The isomorphism $V^{int}_{R+}(h)\cong V^{int}_{R-}(h)$ is given by the zero mode $G_0^V$ of the $\CN=1$ supercurrent. 
	For null momentum $k$, it is convenient to perform a case-by-case analysis. For $F_{24}$, since $V_{R-}(\frac{1}{2})=V_{R+}(\frac{1}{2})=0$, we simply get 
	\be H^1_{-1/2}(k)=0\qquad\qquad   \text{for $k^2=0$, \ \ $V^{int}=F_{24}$}\ .\ee For the $V^{fE_8}$ and $V^{f\natural}$ cases, we refer to \cite{Sch1} and \cite{Harrison:2018joy}, respectively.
	In superstring theory, the isomorphisms for $k\neq0$ are the statements that the BRST quantization for non-zero momentum is equivalent to the light-cone quantization. This isomorphism is actually an isometry, since it preserves the bilinear forms on the cohomology groups and on the SVOA (and its module).
	\item[Cohomology representatives from old covariant quantization] A particularly useful set of representatives for the cohomology classes in $H^1_{-1}(k)$ is given  by states of the form
	\be\label{OCQstates} c_1e^{-\phi} |\chi,k\rangle\ ,
	\ee where $|\chi,k\rangle$ is a state of momentum $k$ in the matter SVA $V^m_{NS}=V^{int}_{NS}\otimes V^{X,\psi}_{NS}$ that satisfies
	\begin{align}\label{OCQ1} &(L^m_0-\frac{1}{2}) |\chi,k\rangle=0\ ,\\
\label{OCQ2} &L^m_n|\chi,k\rangle=G^m_r|\chi,k\rangle=0\ ,\qquad\qquad \forall n,r>0\ ,
	\end{align} i.e. it is a superconformal primary of weight $1/2$. It is easy to see that states of the form \eqref{OCQstates} and satisfying \eqref{OCQ1},\eqref{OCQ2} are $Q$-closed and therefore define classes in $H^1_{-1}(k)$. Vice versa it can be shown \cite{Polchinski1}  that every class in $H^1_{-1}(k)$ has a representative of this form, but possibly more than one (i.e. some of the states \eqref{OCQstates} might be $Q$-exact). Similarly, all classes in $H^1_{-1/2}(k)$ admit (possibly non-unique) representatives of the form
	\be c_1e^{-\phi/2} |u,k\rangle
	\ee where $|u,k\rangle$ is a state of momentum $k$ in the matter Ramond sector $V^m_{R}=V^{int}_{R}\otimes V^{X,\psi}_{R}$ with fermion number $-1$ and such that 
	\begin{align}
	\label{OCQ2R} &L^m_n|u,k\rangle=G^m_r|u,k\rangle=0\ ,\qquad\qquad \forall n>0,r\ge 0\ ,
	\end{align} (since $(G^m_0)^2=L^m_0-\frac{5}{8}$, the condition $G^m_0|u,k\rangle=0$ implies $(L^m_0-\frac{5}{8})|u,k\rangle=0$.) The definition of the space of physical states in terms of states satisfying \eqref{OCQ1},\eqref{OCQ2} or \eqref{OCQ2R} is known as the `Old Covariant Quantization' in superstring theory.
	\item[Zero momentum] It is clear from the previous observations that the $k=0$ sector needs to be considered separately, since most of the theorems we mentioned above do not apply in this case. Fortunately, it is easy to find the cohomology groups by a direct computation. The analysis is described in appendix \ref{a:details}. The outcome is that in the $(-1)$-picture (NS sector) the cohomology is non-zero only at degrees $0,1,2$, with \be \dim H^0_{-1}(0)=1
	\ee \be \dim H^1_{-1}(0)=\dim V^{int}_{NS}(\frac{1}{2})+2\ee 
	\be \dim H^2_{-1}(0)=1.
	\ee In fact, all states in $C_{-1}^i(0)$ are $Q$-closed and there are no $Q$-exact states, so that there is an isomorphism of $C_{-1}^i(0)\cong H^i_{-1}(0)$. In particular, $H^1_{-1}(0)$ is spanned by
	\be \psi^\mu_{-1/2}e^{-\phi}c_1|0\rangle\ ,\qquad\qquad \mu\in\{+,-\}\ee
	\be\label{zeromomstates} v^a_{-1/2}e^{-\phi}c_1|0\rangle\ ,\qquad i=1,\ldots,N=\dim V^{int}_{NS}(\frac{1}{2})\ ,\ee
where $v^a$, $a=1,\ldots, N$ are the fields of weight $1/2$ in $V^{int}$. One has $N=24$ for $F_{24}$, $N=8$ for $V^{fE_8}$ and $N=0$ for $V^{f\natural}$. Notice that these states are of the form \eqref{OCQstates}.   \\
In the $(-1/2)$-picture, the cohomology is non-zero only at ghost number $1$, with
\be \dim H^1_{-1/2} (0)=\dim V^{int}_R(\frac{1}{2})\ .
\ee The description of Ramond states is a bit more complicated, and we refer to \cite{Sch1} and \cite{Harrison:2018joy} for the cases $V^{int}=V^{fE_8}$ and $V^{int}=V^{f\natural}$. When $V^{int}=F_{24}$, one has $\dim V^{int}_R(\frac{1}{2})=0$, so there is no cohomology for $k=0$ at picture number $-1/2$ or $-3/2$.
\end{description}
To conclude, the space of physical states is given by
\be \CH_{phys} =\oplus_{k\in \Gamma^{1,1}} (H^1_{-1}(k)\oplus H^1_{-1/2}(k))\ .
\ee
Notice that the dimensions of the cohomology spaces do not depend on the choice of the $\CN=1$ structure on $V^{int}$, in particular when $V^{int}$ is $F_{24}$. However, the representatives of the cohomology classes do depend on this choice and, most importantly, the superalgebra of physical states that we will define in the next section depends on this choice.

\subsection{Lie superalgebra of physical states}\label{sec:LiePhysStates}

Let us now exhibit the structure of Lie superalgebra on the space $\CH_{phys}$ of physical states. The starting point is to define a Lie superalgebra structure on the graded complex $C$. Following \cite{LZ2}, we define a Lie bracket $\{,\}:C^n_p(k)\times C^m_q(k')\to C^{n+m-1}_{p+q}(k+k')$ by
\be \{u,v\}=(-1)^{|u|}(b_{-1}u)_0v
\ee where the parity $|u|\in \ZZ/2\ZZ $ of an element $u\in C^n_p(k)$ is defined by $|u|=n+2p+1\mod 2$. This Lie bracket satisfies $\ZZ_2$-graded versions of skew-symmetry and Jacobi identity, and is compatible with the picture changing operator $X$ and BRST charge $Q$, in the sense that
\be X\{u,v\}=\{Xu,v\}=\{u,Xv\}\ ,
\ee
\be Q\{u,v\}=\{Qu,v\}+(-1)^{|u|+1}\{u,Qv\}\ .
\ee The latter property ensures that $\{,\}$ induces a well-defined bracket (which, by slight abuse of notation, we denote by the same symbol) $\{,\}:H^n_p(k)\times H^m_q(k')\to H^{n+m-1}_{p+q}(k+k')$ on the BRST cohomology -- the bracket between $Q$-closed states is still $Q$-closed, and the bracket between a $Q$-exact and a $Q$-closed state is $Q$-exact.

One then defines a Lie superalgebra $\g=\g_0\oplus \g_1$, where the even and the odd components $\g_0$ and $\g_1$ are given, respectively, by the Neveu-Schwarz and by the Ramond physical states:
\be \g_0=\bigoplus_{k\in \Gamma^{1,1}} H^1_{-1}(k)\qquad \g_1=\bigoplus_{k\in \Gamma^{1,1}} H^1_{-1/2}(k).
\ee The Lie bracket $[u,v]$ on classes $u\in H^1_p(k)$ and $v\in H^1_q(k')$  is defined by
\be [u,v]=\begin{cases} \{u,v\} \in H^1_{-1}(k+k')& \text{if }p=q=-1/2\ ,\\
	X\{u,v\}\in H^1_{p+q+1}(k+k') & \text{otherwise}\ .
	\end{cases}
\ee In other words, when both $u$ and $v$ are odd (in the Ramond sector), then the bracket $[,]$ coincides with $\{,\}:H^1_{-1/2}(k)\times H^1_{-1/2}(k')\to H^1_{-1}(k+k')$; when one of the elements (say, $u$) is even (in the NS sector), then one needs first to map it to its $0$-picture version $Xu\in H^1_0(k)$ and then use the bracket $\{,\}:H^1_0(k)\times H^1_p(k')\to H^1_p(k+k')$. In particular, the picture changing operator $X$, and therefore the bracket $[\,,\,]$, depends on the choice of the $\CN=1$ supercurrent.

This is made more explicit if we take representatives of $H^1_{-1}(k)$ of the form \eqref{OCQstates}, i.e. $u=c_1e^{-\phi}|\chi,k\rangle$, where $|\chi,k\rangle$ is a state in the matter vertex algebra $V^m_{NS}=V^{int}_{NS}\otimes V^{X,\psi}_{NS}$ that is a superconformal primary of weight $1/2$ (see eqs.\eqref{OCQ1},\eqref{OCQ2}). Then
\be Xu=Xc_1e^{-\phi}|\chi,k\rangle=c_1G^m_{-1/2}|\chi,k\rangle+\gamma_{1/2}|\chi,k\rangle
\ee so that
\be b_{-1}Xu=G^m_{-1/2}|\chi,k\rangle\ .
\ee As a consequence, $(b_{-1}Xu)_0$ is just the zero mode of the current corresponding to the weight $1$ matter state $G^m_{-1/2}|\chi,k\rangle$.

The bilinear form $( \cdot, \cdot)_H$ that non-degenerately pairs $H^{n}_p(k)$ with $H^{2-n}_{-2-p}(-k)$ determines a non-degenerate bilinear form  $\langle \cdot|\cdot\rangle$ on $\g_0\oplus \bigoplus_{k\neq 0}\g_1(k)$, defined by
\be \langle u|v\rangle=\begin{cases}
		(u,v)_H& \text{if }u,v\in \g_0\ ,\\
	-(\tilde u,v)_H\quad \text{with }u=X\tilde u  & \text{if }u,v\in \bigoplus_{k\neq 0}\g_1(k)\ ,\\
	0 & \text{otherwise}\ .
\end{cases}
\ee On $\g_1(0)=H^1_{-1/2}(0)$ the bilinear form is in general not defined, because the picture changing operator $X:H^1_{-3/2}(0)\to H^1_{-1/2}(0)$ is not an isomorphism in this case. In \cite{Sch1} and \cite{Harrison:2018joy}, it was proven that when $V^{int}$ is $V^{fE_8}$ or $V^{f\natural}$, $\g_0(0)\oplus \bigoplus_{k\neq 0}\g_1(k)$ is a subalgebra of $\g$, and in particular it is the derived subalgebra $[\g,\g]$. When $V^{int}=F_{24}$, one has $\g_1(0)=0$, so this case is simpler. This form is symmetric when restricted to $\g_0$ and antisymmetric when restricted to $\g_1$ (see \cite{Sch1},  proposition 5.17 for a proof); bilinear forms on a superspace satisfying this property are called supersymmetric. The form $\langle \cdot|\cdot\rangle$ is also invariant, meaning that $\langle [w,x]|y\rangle=\langle x|[w,y]\rangle$ for all $x,y,w\in \g$; this properties follows from analogous properties of the bilinear form on the vertex algebras.

\subsubsection*{Cartan subalgebra and root multiplicities} Let us specialize to the case where $V^{int}\cong F_{24}$, and consider the even $k=0$ component $\g_0(0)$, which is a finite dimensional Lie subalgebra of  $\g$ (again, $\g_1(0)=0$). Acting by $b_{-1}X$ on the two states $\psi^\pm_{-1/2}e^{-\phi}c_1|0\rangle$,  we get the weight $1$ states $G^m_{-1/2}\psi^\pm_{-1/2}|0\rangle=\alpha^\pm_{-1}|0\rangle$ corresponding to the space-time currents $\partial X^\pm$. The zero modes are $P^\pm= \alpha_0^\pm$, whose eigenvalues are the space-time momenta $k^+,k^-$. These operators obey the commutation relations
\be [P^\mu,u]=k^\mu u\qquad\qquad u\in \g(k)\ ,
\ee with elements $u\in \g$ of definite momentum $k$. An obvious consequence is that the only generators commuting with both $P^+$ and $P^-$ are the ones in the zero momentum component $\g(0)$. When the internal SVOA is $F_{24}$, there are $24$ further states of the form \eqref{zeromomstates} in $H^1_{-1}(0)$. The $0$-picture version of these states correspond to the $24$ currents $J^a$, $a=1,\ldots,24$, that are superconformal descendants of the weight $1/2$ fields $\lambda^a$. As described in section \ref{s:F24}, the zero modes of these currents generate a semi-simple Lie algebra $g\subset so(24)$ of dimension $24$. Thus, the zero momentum subalgebra $\g_0(0)$ of $\g$  is isomorphic to
\be \g_0(0)=u(1)^{\oplus 2}\oplus g\ ,
\ee with the abelian component $u(1)^{\oplus 2}$ generated by $P^+,P^-$. A maximal abelian subalgebra of $\g_0(0)$ is given by \be \h=u(1)^{\oplus 2}\oplus h\ ,\ee where $h\subset g$ is a Cartan subalgebra for $g$, with generators $\alpha_1^\vee,\ldots, \alpha_r^\vee$. For $V^{int}=F_{24}$, the zero momentum odd component $\g_1(0)$ is $0$, and no other generator with nonzero momentum can commute with $P^\mu\in \h$. We conclude that $\h$ is actually a maximal abelian subalgebra for the whole $\g$. Thus, $\g$ has rank $r+2$, where $r$ is the rank of $g$. We see that, while our construction provided a natural $\Gamma^{1,1}$-grading for $\g$ in terms of momentum (i.e., $P^\mu$ eigenvalues), by taking into account the eigenvalues with respect to the remaining $r$ generators $\alpha_1^\vee,\ldots, \alpha_r^\vee$ of the Cartan algebra, we can now introduce a finer grading for the superalgebra $\g$ with values in the lattice
\be \Qg_g:=\Gamma^{1,1}\oplus \tilde Q_g\cong \ZZ\oplus\ZZ\oplus \tilde  Q_g \subset \h^*\ .\ee Here, $\tilde Q_g:=Q_g\cup(\rho+Q_g)\subseteq P_g$ is the union of the root lattice $Q_g$ of the finite dimensional algebra $g$ and its translate $\rho+Q_g$ (see eq.\eqref{rooti}) . Thus
\be \g=\oplus_{\hat k\in \Qg_g} \g(\hat k)\ ,
\ee where $\hat k=(m,n,w)\in \ZZ\oplus \ZZ\oplus \tilde Q_g$. In particular, the even and odd components are graded as
\be \g_0=\bigoplus_{\substack{m,n\in \ZZ\\ w\in Q_g}} \g_0(m,n,w)\ ,
\ee and
\be \g_1=\bigoplus_{\substack{m,n\in \ZZ\\ w\in \rho+Q_g}} \g_1(m,n,w) \ .
\ee
The bilinear form on the cohomology, when restricted to the zero momentum space NS space, defines an invariant non-degenerate bilinear form $\langle \cdot |\cdot \rangle_{\g}$ on $\g_0(0)$ which extends the Cartan-Killing form $(\cdot |\cdot )_g$ on $g$. This form satisfies
\be \langle P^+|P^-\rangle=-1\ ,\qquad\qquad \langle P^+|P^+\rangle=\langle P^-|P^-\rangle=0
\ee and both $P^+$ and $P^-$ are orthogonal to $g$. With this choice, one has
\be \hat k^2\equiv \langle m,n,w|m,n,w\rangle=-2mn+(w|w)_g\ ,
\ee for $\hat k\equiv (m,n,w)\in \Qg_g$. We denote by 
\be \hat \Delta\equiv \hat \Delta_0\cup \hat \Delta_1:=\{\hat k\in \Qg_g\mid \g(\hat k)\neq 0\} \subset \Qg_g\ ,\ee the set of roots of $\g$, with $\hat \Delta_0$ and $\hat \Delta_1$ the subsets of even and odd roots, respectively. The bilinear form $\langle\cdot|\cdot\rangle$ restricted to the real spaces \be \h_\RR:=\RR P^+\oplus \RR P^-\oplus (\rootQ^\vee_g\otimes \RR)\ ,\ee and $\h^*_{\RR}:=\Qg_g\otimes \RR$ is real-valued, with signature $(r+1,1)$.

Since all (super)ghosts and superconformal generators commute with $\alpha_1^\vee,\ldots, \alpha_r^\vee$, the equivalence between BRST and light-cone quantization is compatible with this finer grading. As a consequence, if we denote by
\be V^{int}_{NS}(n,w)\qquad\qquad  n\in \frac{1}{2}\ZZ,\ w\in Q_g\ ,
\ee the component of the SVOA $V^{int}_{NS}\cong F_{24}$ with $L_0$-eigenvalue $n$ and $\alpha_1^\vee,\ldots, \alpha_r^\vee$-eigenvalues $w$ and
\be V^{int}_{R}(n,w)\qquad\qquad  n\in \ZZ,\ w\in \rho+Q_g\ ,
\ee the analogous component in the twisted module $V^{int}_R$, one has
\be \dim \g_0(m,n,w)=\dim V^{int}_{NS-}(mn,w)=c_{NS-}(nm,w)\ ,
\ee
\be \dim \g_1(m,n,w)=\dim V^{int}_{R\pm }(mn,w)=c_{R+}(nm,w)=c_{R-}(nm,w)\ .
\ee Here, $c_{NS-}(nm,w)$ and $c_{R\pm}(nm,w)$ are the Fourier coefficients of the Jacobi forms \eqref{Jaceven} and \eqref{Jacodd}. As discussed in appendix \ref{a:multiJacobi} and section \ref{s:partfunct}, general properties of the coefficients of Jacobi forms imply that the dimension of the root spaces $\g_0(m,n,w)$ and $\g_1(m,n,w)$ depend only on the norm $-2mn+(w|w)_g$ of the root and on the class $[w]$ of $w$ in the quotient $P_g/i(\rootQ_g^\vee)$. Furthermore,  the condition  \eqref{coeffboundB} implies that
\be\label{rootbound} \dim \g(m,n,w)\neq 0\qquad \Rightarrow\qquad \begin{cases}
	mn\ge 0\\ -2mn+(w|w)_g\le M\ ,
\end{cases}
\ee where $M>0$ is a constant depending on the choice of the $\CN=1$ superalgebra: see \eqref{Mbound}.

\section{The Lie superalgebra $\g$ as a Borcherds-Kac-Moody superalgebra}\label{sec:BKMproof}

In this section we will prove that the Lie superalgebra $\g$ that we constructed in the previous section is a Borcherds-Kac-Moody (BKM) superalgebra. BKM algebras differ from the usual Kac-Moody algebras because the \emph{simple} roots are allowed to have non-positive norm. They can be defined in terms of Chevalley-Serre generators and relations (see for example \cite{Ray}). In our case, it is useful to use an  alternative characterization of BKM superalgebras, which was given by Ray \cite{Ray}, and we begin by describing this below before embarking on the proof. 

\subsection{Generalities on BKM superalgebras}

First, we list some relevant definitions. According to definition 2.3.17 of \cite{Ray}, a root $\alpha\in \hat\Delta$ is said to be of \emph{finite type} if it acts locally nilpotent on $\g$, i.e.  if for all $x\in \g(\alpha)$ and for all $y\in \g$, there is an integer $n$ (possibly depending on $x$ and $y$), such that $({\rm ad}\,x)^ny=0$. A root is said to be of \emph{infinite type} if it is not of finite type. The bound \eqref{rootbound} on the norm of the roots implies that a root of positive norm is necessarily of finite type. Indeed, if $\alpha\in \hat \Delta$ with $\langle \alpha|\alpha\rangle>0$, then for any $\beta\in \hat\Delta$ we have
\be \langle \beta+n\alpha|\beta+n\alpha\rangle = \langle \beta|\beta\rangle+2n\langle \alpha|\beta\rangle+n^2\langle \alpha|\alpha\rangle\stackrel{n\to \pm\infty}{\longrightarrow} +\infty\ .
\ee Thus, for sufficiently large $n$, $\beta+n\alpha$ is not a root, so that $({\rm ad}\,x)^ny=0$ for all $x\in \g(\alpha)$ and $y\in \g(\beta)$.

\begin{theorem}[\cite{Ray}, corollary 2.5.11]
Let $G=G_0\oplus G_1$ be a (complex) Lie superalgebra. Suppose that the following conditions are satisfied:
\begin{enumerate}
	\item There is a self-centralizing even subalgebra $H\subset G$ such that $G$ can be decomposed as a direct sum $\oplus_{\alpha} G_\alpha$ of eigenspaces for $H$, with each eigenspace $G_\alpha$ being finite dimensional. A non-zero eigenvalue $\alpha\in H^*$ is called a \emph{root} of $G$.
	\item There is a non-degenerate, supersymmetric, invariant bilinear form $\langle\cdot |\cdot\rangle$ on $G$, with respect to which $G_0$ and $G_1$ are orthogonal to each other.
	\item The algebra $H$ admits a real form $H_\RR$ such that the restriction of $\langle\cdot |\cdot\rangle$ to $H_\RR$ is real (so that $H_\RR\cong H_\RR^*$). Furthermore, $H_\RR^*\cong H_\RR$ contains all roots.
	\item  There is an element $h\in H_\RR$ (a regular element) that is not orthogonal to any root and such that for all $N>0$ there is only a finite number of roots $\alpha$ such that $0<|\alpha(h)|<N$. A root is called positive if $\alpha(h)>0$ and negative if $\alpha(h)<0$.
	\item For any $\alpha,\beta$ of infinite type or of zero norm that are both positive or both negative, one has $\langle \alpha|\beta\rangle\le 0$. Moreover, if $\langle \alpha|\beta\rangle=0$ and if $x\in G_\alpha$ is such that $[x,G_{-\gamma}]=0$ for all roots $\gamma$ with $|\gamma(h)|<|\alpha(h)|$, then $[x,G_\beta]=0$.
\end{enumerate}
Then, $G$ is a Borcherds-Kac-Moody superalgebra.
\end{theorem}

\subsection{Proof that $\mathfrak{g}$ is a BKM superalgebra}

Using the characterization of BKM superalgebras presented in the previous subsection, we can prove that the Lie superalgebra $\g$, constructed in \S\ref{sec:LiePhysStates}, is a BKM superalgebra. The following lemma will prove to be a useful intermediate step. 
\begin{lemma}\label{l:regular}
	Let $\eta \in \rootQ_g^\vee \subset h$ be an element of the Cartan subalgebra of $g$ such that $\alpha(\eta)\neq 0$ for all non-zero roots $\alpha\in \Delta_g$ of $g$. Then, there exists a positive integer $L$ such that the element $\hreg=-LP^+-LP^-+\eta \in \h_\RR$ in the real Cartan subalgebra of $\g$ satisfies the following properties:
	\begin{enumerate}
		\item if $\gamma=(m,n,w)\in \hat\Delta_g$ is a non-zero root of $\g$, then $\gamma(\hreg)\neq 0$;
		\item if $\gamma=(m,n,w)\in \hat\Delta_g$ with $m>0$ or $n>0$, then $\gamma(\hreg)> 0$;
		\item for all $N>0$, there are only a finite number of roots $\gamma\in \hat\Delta_g$ such that $0<|\gamma(\hreg)|<N$;
		\item if $\alpha=(m,n,w)\in \hat\Delta_g$, with $\langle \alpha|\alpha\rangle\equiv -2mn+(w|w)=0$, is a non-zero null root of $\g$, and $\gamma=(0,0,w')\in \hat \Delta_g$ is a root with $\gamma(P^+)=\gamma(P^-)=0$, then $|\alpha(\hreg)|>|\gamma(\hreg)|$.
	\end{enumerate}
\end{lemma}
\begin{proof}
	We take $L$ to be very large, so that, in particular, 
	\be \langle \hreg|\hreg\rangle= -L^2+(\eta|\eta)<0\ .
	\ee Let us prove the $\hreg$ is not orthogonal to any root, for $L$ large enough.  If $\gamma=(m,n,w)\in \hat\Delta$ is a root with $(m,n)\neq (0,0)$, then by \eqref{rootbound} one has
	\be w(\eta)^2\le (\eta|\eta)(w|w)\le (\eta|\eta)(M+2mn)
	\ee so that
	\be L^2(m+n)^2-w(\eta)^2\ge 
	L^2(m^2+n^2)+2mn(L^2-(\eta|\eta))-(\eta|\eta)M\ge L^2-(\eta|\eta)M>0
	\ee where we used that $mn\ge 0$ by \eqref{rootbound}, that $L^2-(\eta|\eta)=-\langle\hreg|\hreg\rangle>0$, that $m^2+n^2\ge 1$ for $m,n$ not both null, and that for $L$ large enough $L^2>(\eta|\eta)M$. This means that a root $\gamma=(m,n,w)$ of $\g$ with $(m,n)\neq (0,0)$ is positive $\gamma(\hreg)=L(m+n)+w(\eta)>0$ if and only if $m,n\ge 0$ (notice the $m$ and $n$ cannot have opposite sign, since for a root $mn\ge 0$). If $\gamma=(m,n,w)$ is a root of $\g$ with $m=n=0$, then $w$ must be a root of $g$; then, $\gamma(\hreg)>0$  if and only if $w(\eta)>0$, i.e. if $w$ is a positive root of $g$. This shows that no root is orthogonal to $\hreg$. 
	
 Without  loss of generality, we can assume that $\hreg$ is a primitive vector in the lattice $\ZZ P^+\oplus \ZZ P^-\oplus \rootQ^\vee\subset \h_\RR$. If we denote by $\Pg=\Gamma^{1,1}\oplus P_g\subset \h^* $ the dual lattice, then there exists $u\in \Pg$ such that $u(\hreg)=1$. Any root $\gamma\in \hat \Delta_g\subset \Qg\subseteq \Pg$ can be uniquely decomposed as $tu+\gamma_\perp$, where $t=\gamma(\hreg)\in \ZZ$ and $\gamma_\perp\in \Pg\cap \hreg^\perp$. By \eqref{rootbound} $\langle \gamma|\gamma\rangle\le M$, so that for each fixed $t\in \ZZ$ there is an upper bound $B(t)>0$ such that $\langle \gamma_\perp|\gamma_\perp\rangle\le B(t)$ for all $ut+\gamma_\perp\in \hat\Delta$. Since $\Pg\cap \hreg^\perp$ is a positive definite lattice, for each $t\in \ZZ$ there are only finitely many $\gamma_\perp\in \Pg\cap \hreg^\perp$ satisfying this bound, and therefore finitely many roots with $\gamma(\hreg)=t$. This proves point 3.

	 As for point 4, it is sufficient to prove it for $\alpha=(m,n,w)\in \hat \Delta$ a null root with $\alpha(\hreg)>0$. Suppose first that $w\neq 0$, so that $2mn=(w|w)\neq 0$.  Let $\gamma=(0,0,w')$ be another non-zero root of $\g$, where $w'$ is a root of $g$. Let us prove that $|\gamma(\hreg)|<\alpha(\hreg)$ for sufficiently large $L$. We have
	\be \alpha(\hreg)=L(m+n)+w(\eta)\ge L(m+n)-\sqrt{(w|w)(\eta|\eta)}= L(m+n)-\sqrt{2mn(\eta|\eta)}\ .
	\ee Set $y=\sqrt{\frac{m}{n}}$ (recall that $mn\neq 0$), so that
	\be \alpha(\hreg)\ge n[L(y^2+1)-\sqrt{2(\eta|\eta)}y]\ .
	\ee As a function of $y$, the right-hand side has a minimum at $y=\frac{\sqrt{2(\eta|\eta)}}{2L}$ with value $nL(1-\frac{(\eta|\eta)}{2L})$, so that
	\be
	\alpha(\hreg)\ge nL(1-\frac{(\eta|\eta)}{2L})\ge L(1-\frac{(\eta|\eta)}{2L})>0\ .  \ee Since there are only finitely many roots of the form $\gamma=(0,0,w')$, $w'\in\Delta_g$, one can choose $L$ sufficiently large so that
	\be |\gamma(\hreg)|=|w'(\eta)|<L(1-\frac{(\eta|\eta)}{2L})\le \alpha(\hreg)\ ,
	\ee for all roots $w'$ of $g$.  Now, suppose that $\alpha=(m,n,w)\neq 0$ is a positive null root with $w=0$. This implies $mn=0$, so that either $m=0$ or $n=0$, but not both. Thus, for sufficiently large $L$, we have
	$$\alpha(\hreg)=L(m+n)\ge L>|w'(\eta)|=|\gamma(\hreg)|\ ,$$
	for all roots $w'$ of $g$.
\end{proof}
We are now ready to prove the main theorem:
\begin{theorem}\label{mainthm}
	$\g$ is a BKM superalgebra.
\end{theorem}
\begin{proof}
	The subalgebra $\h$ constructed in the previous subsection is a self-centralizing even subalgebra, and all components $\g(\alpha)$ in the decomposition $\oplus_{\alpha\in \Qg_g} \g(\alpha)$ are finite dimensional. As a real form $\h_\RR$, we can take the real algebra generated by $P^\mu$, $\mu\in\{+,-\}$ and by the coroots $\alpha^\vee_1,\ldots, \alpha^\vee_r$; the latter generate a Cartan subalgebra of the compact real form of the finite dimensional Lie algebra $g$. The dual space $\h^*_\RR$ contains the root lattice $\Qg=\Gamma^{1,1}\oplus \tilde Q_g$, and therefore all roots of $\g$.    The non-degenerate bilinear form satisfies all the required properties: $\langle \cdot|\cdot\rangle$ is non-degenerate, supersymmetric, invariant, and $\g_0$ is orthogonal to $\g_1$. Its restriction to $\h_\RR$ is real with signature $(r+1,1)$. Eq.\eqref{rootbound} implies that the norms of the roots are bounded from above. As a regular element, we can take an element $\hreg\in \h_\RR$ as in Lemma \ref{l:regular}, which clearly satisfies the properties in point 4. 
	 
	To complete the proof, we just need to establish point 5. As discussed above, a root of infinite type in $\g$ cannot have positive norm. For a lattice of Lorentzian signature, if $\alpha,\beta$ are both positive or both negative of non-positive norm, then they belong to the same connected component of the cone of non-positive norm vectors, so that their product automatically satisfies $\langle \alpha|\beta\rangle\le 0$. Furthermore, one has $\langle \alpha|\beta\rangle=0$ if and only if $\alpha$ and $\beta$ are both null and are proportional to each other. Let  $\alpha=(m,n,w)\in \hat\Delta$ be any non-zero null root of $\g$. By Lemma \ref{l:regular}, any $\gamma=(0,0,w')\in \hat\Delta$ satisfies $|\gamma(\hreg)|<|\alpha(\hreg)|$. Let us prove that a non-zero element $x\in \g(\alpha)$ cannot commute with $\g(0,0,-w')$ for all $0\neq w'\in\Delta_g$. If $w\neq 0$, then $x$ belongs to a non-trivial representation of the finite Lie algebra $g$, so it cannot commute with all generators of $g_+$ and $g_-$.
When $w=0$, since $\alpha=(m,n,0)$ is null and non-zero, one has that either $m=0$ or $n=0$, but not both. For such $m,n$, one has that $(m,n,w)$ is a root of $\hat g$ if and only if $w$ is a root of $g$, and $\oplus_{w\in \Delta} \g(m,n,w)$ forms a $24$-dimensional adjoint representation of $g$. This means that no non-zero element $x\in \g(\alpha)\subset \oplus_{w\in \Delta} \g(m,n,w)$ can commute with all $\g(-\gamma)$ for all $\gamma$ of the form $(0,0,w')$. 
\end{proof}

\subsection{Simple roots and Weyl vector}

In this section, we discuss some of the simple roots of the BKM algebras $\g$ and the existence of a Weyl vector. A complete description of all simple roots of $\g$ requires a case by case treatment. In section \S\ref{sec:A1^8} we perform this analysis for the BKM algebra corresponding to $A_1^8$, while we leave the other cases to future work.
\begin{proposition}
	Let $\alpha_1,\ldots,\alpha_r\in \Delta_g$ be the simple roots of $g$. If $g$ is the sum $g=\oplus_{k=1}^ng_k$ of $n$ simple components $g_k$, $k=1,\ldots,n$, let $\theta_k\in\Delta_g$ be the highest root of $g_k$. Then $\hat \alpha_i:=(0,0,\alpha_i)$, $i=1,\ldots,r$, $\delta_k^+:=(1,0,-\theta_k)$, $\delta_k^-=(0,1,-\theta_k)$, $k=1,\ldots,n$, are real simple roots of $\g$. For each $k=1,\ldots, n$, let $I_k\subseteq \{1,\ldots,r\}$ be such that $\{\alpha_i\}_{i\in I_k}$ is the set of simple roots of $g_k$, and set $D^+_k:=\{ \hat\alpha_i\}_{i\in I_k}\cup \{\delta_k^+\}$  and $D^-_k:=\{\hat\alpha_i\}_{i\in I_k}\cup \{\delta_k^-\}$. Then, the subalgebra of $\g$ generated by $\bigoplus_{\pm\gamma\in D^+}\g(\gamma)$ and the subalgebra generated by $\h\oplus\bigoplus_{\pm\gamma\in D^-}\g(\gamma)$ are both isomorphic to the affine Kac-Moody algebra $\hat g_k$.
\end{proposition}
\begin{proof}
	A root $\alpha=(m,n,w)$ is positive if $m,n\ge 0$ and, in the case $m=n=0$, if $w$ is a positive root of $g$, i.e. $w\in \Delta_g^+$. Therefore, $\hat \alpha_i:=(0,0,\alpha_i)$, $i=1,\ldots,r$ are necessarily simple. The space $\g(1,0):=\oplus_{w\in P_g} \g(1,0,w)$ is $24$ dimensional and transforms in the adjoint representation of $g$, so $(1,0,w)$ is in $\hat\Delta$ if and only if $w\in \Delta_g$. The only way to obtain a root of the form $(1,0,w)$ as a sum over positive roots  is as  $(1,0,w)=(1,0,w-w')+(0,0,w')$ where $w'\in\Delta^+_g$ and $w-w'\in\Delta_g$. But if $w=-\theta_k$, then $(1,0,-\theta_k-w')$ is not in $\hat\Delta$ for any $w'\in \Delta_g^+$, so $(1,0,-\theta_k)$ must be a simple root. An analogous result holds for roots of the form $(0,1,-\theta_k)$. For the last statement, it is sufficient to notice that, if $\{\gamma_1,\gamma_2,\ldots\}$ is a set of real simple roots equal to either $D^+_k$ or $D^-_k$, then  the matrix $(A_{ij})=\langle \gamma_i|\gamma_j\rangle$  is the Cartan matrix of the affine Kac-Moody algebra $\hat g_k$, so the subalgebra generated by the corresponding root elements must be isomorphic to $\hat g_k$ .\end{proof}

We stress that, while the simple real roots $\hat \alpha_i:=(0,0,\alpha_i)$, $i=1,\ldots,r$, $\delta_k^+:=(1,0,-\theta_k)$, $\delta_k^-=(0,1,-\theta_k)$, $k=1,\ldots,n$, span the space $\h^*$, this does not necessarily mean that they form a complete set of real simple roots. For example, in section \S\ref{sec:A1^8}, we will show that in the case $g=A_1^8$ there are infinitely many real simple roots.

Let $\rho=\sum_{k=1}^n\rho_k$ be the Weyl vector of the algebra $g=\oplus_{k=1}^n g_k$, with $\rho_k$ the Weyl vector of the simple component $g_k$. The Weyl vector obeys the usual property
\be (\rho|\alpha_i)=\frac{1}{2}(\alpha_i|\alpha_i)\ .
\ee 
Furthermore, with the normalization we have chosen for the Killing form, we obtain
\be (\theta_k|\theta_k)=\frac{2}{\dcox_{g_k}}\ ,\qquad (\rho|\theta_k)=(\rho_k|\theta_k)=1-\frac{1}{\dcox_{g_k}}\ .
\ee Thus, if we define $\hat \rho=(-1,-1,\rho)\in \Qg_g$, we get
\be \langle \hat\rho|\hat\alpha_i\rangle=(\rho|\alpha_i)=\frac{1}{2}\langle \alpha_i|\alpha_i\rangle\ ,
\ee and
\be \langle \hat\rho|\delta_k^\pm\rangle=1-(\rho|\theta_k)=\frac{1}{\dcox_{g_k}}=\frac{1}{2}\langle \delta_k^\pm|\delta_k^\pm\rangle\ .
\ee  The condition that $\langle \hat\rho|\alpha\rangle=\frac{1}{2}\langle \alpha|\alpha\rangle$ for all simple roots $\alpha$ is the defining property of a Weyl vector for the algebra $\g$. Since the space $\h^*$ is spanned by the simple real roots $\hat\alpha_1,\ldots,\hat\alpha_r$, $\delta^+_k$, $\delta^-_k$, we conclude that if the algebra $\g$ admits a Weyl vector, then it must be equal to
\begin{equation}
\hat \rho=(-1,-1,\rho)\in \Qg_g.
\end{equation}
To verify that this is actually the Weyl vector of the algebra, one must  check that it satisfies the defining properties with respect to \emph{all} the real and imaginary simple roots of $\g$. In \S\ref{s:a18roots} we prove that the BKM algebra corresponding to $g=A_1^8$ coincides with an algebra studied by Borcherds  \cite{BorcherdsFake,Borcherds96,BorcherdsMM,Borcherds98}, who showed that the Weyl vector is indeed $\hat \rho$. We conjecture that the Weyl vector exists, and therefore coincides with $\hat\rho$, for all the other cases as well, but we leave the proof for future work. Based on this conjecture, in the following we refer to $\hat\rho$ as the Weyl vector of $\g$. Note  that even if $\hat \rho$ only satisfies the defining properties with respect to the real simple roots, it may still be used to construct the denominator formula, as we discuss in \S\ref{sec:denom} below.
By \eqref{Weylnorm}, we obtain
\be \langle \hat\rho|\hat\rho\rangle = -2+ (\rho|\rho)=0\ ,
\ee so that the Weyl vector has zero norm. Finally, we notice that $-\hat \rho=(1,1,-\rho)$ is an odd simple root; this follows from the fact that, for $V^{int}=F_{24}$, one has $H_{-1/2}^1(k)\cong V^{int}_R(-\frac{k^2}{2}+\frac{1}{2})$ is nonzero only for $-\frac{k^2}{2}+\frac{1}{2}\ge 3/2$. Given that $k^2=-2mn$, this means that $mn\ge 1$. For $m=n=1$, one has that $H_{-1/2}^1(k)\cong V^{int}_R(\frac{3}{2})$ which is the sum of irreducible representations of $g$ with lowest weight $-\rho$. This implies that  $(1,1,-\rho)$ cannot be obtained as a sum of positive roots, and therefore it is simple. 

 If $\gamma\in \hat \Delta$ is a root of non-zero norm, it makes sense to consider the reflection $r_\gamma$ with respect to hyperplane perpendicular to $\gamma$
\be r_{\gamma}(\beta)=\beta-\frac{2\langle \gamma|\beta\rangle}{\langle \gamma|\gamma\rangle} \gamma\ , 
\ee where $\beta\in \h^*$. Following \cite{Ray}, we define the Weyl group $W$ of the infinite dimensional BKM algebra $\g$ as the group generated by reflections $r_\gamma$, where the root $\gamma $ is even and real (and therefore automatically of finite type and non-zero norm). As in the finite dimensional case, the Weyl group preserves the bilinear form
\be \langle \mathrm{w}(\alpha)|\mathrm{w}(\beta)\rangle=\langle \alpha|\beta\rangle\ ,\qquad\qquad \mathrm{w}\in W\ ,
\ee for all $\alpha ,\beta\in \h^*$.

Notice the the real roots of $\g$ are always even. Indeed, it is well-known that real roots in BKM algebras have multiplicity $1$ \cite{Ray}. Formulas \eqref{JacRam} and \eqref{Jacodd} show that the multiplicities of all odd roots are a multiple of $2^{r/2-1}$, where $r\ge 4$ is the rank of the algebra $g$. This implies that there are no odd real roots, i.e. all odd roots $\gamma$ satisfy $\langle \gamma|\gamma\rangle \le 0$.

With our knowledge of the roots, we can now study some interesting representation theoretic, automorphic functions associated with our BKM superalgebras $\mathfrak{g}$. 

\subsection{Denominator and superdenominator}\label{sec:denom}

As in the case of finite or Kac-Moody Lie (super)algebras, BKM (super)algebras with Weyl vectors possess a version of the Weyl-Kac character formula which, when one considers the character of the trivial module, produces a Weyl-Kac denominator identity. Each side of the denominator identity contains valuable information about the root spaces, root space multiplicities, and (real and imaginary) simple roots of the algebra in question, and takes the form of  an equality between two very different formulations of a given modular or automorphic object. In certain nice examples, like the Monster BKM, knowledge of both sides of the denominator identity is sufficient to determine the algebra itself. Furthermore, in the ordinary Kac-Moody case, the Weyl-Kac denominator is itself the character for a module whose highest weight is the Weyl vector; again, this holds for some particularly simple low-rank BKMs (see \cite{PPV2} for the proof in the Monster case)\footnote{It is known that this cannot hold for general BKMs since, for instance, there are known examples of BKMs that do not have a Weyl vector.}. 

Let us now write down the (super-)denominator formula for our BKM algebra. First we introduce some generalities. Let $\mathfrak{g}=\mathfrak{g}_0\oplus \mathfrak{g}_1$ be a BKM superalgebra with even and odd components $\mathfrak{g}_0$ and $\mathfrak{g}_1$, respectively. Let us denote the roots by $\alpha\in \hat \Delta$ and set
\begin{equation}
m_0(\alpha)=\dim (\mathfrak{g}_\alpha\cap \mathfrak{g}_0), \qquad m_1(\alpha)=\dim (\mathfrak{g}_\alpha \cap \mathfrak{g}_1)=\text{mult}(\alpha)-m_0(\alpha).
\end{equation}
We further denote the positive even or odd roots by $\hat{\Delta}_{0}^+, \hat{\Delta}_{1}^+$, respectively. Let $I$ be an index set, indexing the simple roots $\alpha_i$. Any root $\alpha$ may then be expanded as $\alpha=\sum_{i\in I} k_i\alpha_i$ and we define the \emph{height} of $\alpha$ to be 
\begin{equation}
ht(\alpha)=\sum_{i\in I} k_i.
\end{equation}
We also define the ``even height'' as 
\begin{equation}
ht_0(\alpha)=\sum_{i\in I \backslash S} k_i,
\end{equation}
where $S\subseteq I$ indexes only the odd roots. Before we can state the denominator formulas we introduce the following sums
\begin{equation}
T=e^{-\hat{\rho}}\sum_{\mu}(-1)^{ht(\mu)} e^{\mu}, \qquad \qquad T'=e^{-\hat{\rho}}\sum_{\mu}(-1)^{ht_0(\mu)} e^{\mu},
\end{equation}
where $\hat\rho$ is the Weyl vector. The sums here are taken over all sums $\mu$ of distinct pairwise orthogonal imaginary simple roots. 

We now have all the ingredients to state the desired formulas. For any super BKM $\mathfrak{g}$ we have the {\it denominator formula} 
\begin{equation} 
\frac{e^{-\hat{\rho}}\, \prod_{\alpha\in \hat{\Delta}_0^+}(1-e^{\alpha})^{m_0(\alpha)}}{\prod_{\alpha\in \hat{\Delta}_1^+}(1+e^{\alpha})^{m_1(\alpha)}}=\sum_{{\rm w}\in W}\det({\rm w}){\rm w}(T),
\end{equation}
and, in addition, we have the {\it super-denominator formula}
\begin{equation}
\frac{e^{-\hat{\rho}}\, \prod_{\alpha\in \hat{\Delta}_0^+}(1-e^{\alpha})^{m_0(\alpha)}}{\prod_{\alpha\in \hat{\Delta}_1^+}(1-e^{\alpha})^{m_1(\alpha)}}=\sum_{{\rm w}\in W}\det({\rm w}){\rm w}(T').
\end{equation}
For obvious reasons we call the left hand side the \emph{product side} and the right hand side \emph{the sum side} of the denominator formula. The sum side is sometimes referred to as the \emph{denominator function}. Thus the denominator formula provides a product representation of the denominator function. 

Let us now discuss the denominator formulas for the super BKM $\mathfrak{g}$ constructed in section~\ref{sec:BKMproof}. Recall from section~\ref{sec:LiePhysStates} that a root $\alpha$ of $\mathfrak{g}$ is parametrized by $(m,n,w)$ and the root multiplicities  are given by 
\begin{equation}
m_0(\alpha)=c_{NS-}(mn,w), \qquad \qquad m_1(\alpha)=c_{R+}(mn, w)=c_{R-}(mn, w),
\end{equation}
where $c_{NS-}(mn,w)$ and $c_{R\pm}(mn, w)$ are the Fourier coefficients of the Jacobi forms $\phi_{NS-}(\tau, \xi)$ and $\phi_{R\pm}(\tau, \xi)$ constructed in section~\ref{s:partfunct}. The Weyl vector of $\mathfrak{g}$ was found in section~\ref{sec:BKMproof} to be 
\begin{equation}
\hat\rho=(-1,-1,\rho)\ .
\end{equation}
 Combining everything, we deduce that the product side of the denominator formula becomes
\be pqe^{-\rho}\prod_{w\in \Delta^+_g}(1-e^w)^{c_{NS-}(0,w)}\prod_{\substack{m,n\in \ZZ_{\ge 0}\\ (m,n)\neq (0,0)} }\prod_{w\in \tilde Q_g} \frac{(1-p^mq^ne^w)^{c_{NS-}(mn,w)}}{(1+ p^mq^ne^w)^{c_{R+}(mn,w)}}.
\ee Here $g$ denotes the underlying finite-dimensional subalgebra of $\mathfrak{g}$ and the zero momentum contribution $e^{-\rho_g}\prod_{\ell>0}(1-e^w)^{c_{NS-}(0,w)}$ coincides with the Weyl denominator formula of $g$. 

Similarly, the product side of the super-denominator formula takes the form
\be pqe^{-\rho}\prod_{w\in \Delta^+_g}(1-e^w)^{c_{NS-}(0,w)}\prod_{\substack{m,n\in \ZZ_{\ge 0}\\ (m,n)\neq (0,0)} }\prod_{w\in \tilde Q_g} \frac{(1-p^mq^ne^w)^{c_{NS-}(mn,w)}}{(1- p^mq^ne^w)^{c_{R-}(mn,w)}}.
\ee 

\section{The example of $g= A_1^8$ }\label{sec:A1^8}

We conclude this note by illustrating the formal properties of our BKMs $\mathfrak{g}$ in the simplest concrete example: when the choice of $\CN=1$ structure in $F_{24}$ produces currents generating the Lie algebra $g= A_1^8$. In particular, we will discuss the root spaces and their multiplicities and the Weyl group of this BKM.

\subsection{Construction}
The $F_{24}$ theory with $\CN=1$ structure of type $A_1^8$ has a symmetry $SU(2)^8\ltimes S_8\subset O(24)$ preserving the $\CN=1$ current. The corresponding finite dimensional Lie algebra $g=su(2)^8$ has dual Coxeter number $\dcox=2$ for all simple components, so that the roots of $g$ have length $\frac{2}{\dcox_g}=1$. The root, coroot and weight lattices are, respectively, $\rootQ=\ZZ^{\oplus 8}$, $\rootQ^\vee=(2\ZZ)^{\oplus 8}$, $P=(\frac{1}{2}\ZZ)^{\oplus 8}$. The Weyl vector is $\rho=(\frac{1}{2},\ldots,\frac{1}{2})$ with norm $(\rho|\rho)=2$, and the highest roots are $\theta_k=(0,\ldots,0,1,0\ldots,0)$ with the $1$ at the $k$-th position, $k=1,\ldots,8$. The even roots $(m,n,w)\in \hat\Delta$ of $\g$ have $w$ valued in the root lattice $Q\subset P_g$, while the odd roots have $w$ in $\rho+\rootQ\subset P_g$. The root lattice of the BKM algebra $\g$ is $\Qg=\Gamma^{1,1}\oplus \tilde Q_g$, where $\tilde Q_g=\rootQ_g\cup (\rho+\rootQ_g)=\ZZ^{\oplus 8}\cup (\frac{1}{2}+\ZZ)^{\oplus 8}$. The dual lattice 
\be\label{D8lattice} \tilde Q_g^*=\{(x_1,\ldots,x_8)\in \ZZ^{\oplus 8}\mid \sum_i x_i\in 2\ZZ\}\cong D_8
\ee is an even lattice isomorphic to the root lattice $D_8$.

\subsection{Description of real roots}
In order to find the multiplicities of the real roots of $\g$, we proceed as follows. If $\gamma=(m,n,w)\in \hat\Delta_0$ is an even root, we know that the multiplicity $m_0(m,n,w)$ is the Fourier coefficient $c_{NS-}(mn,w)$ of the Jacobi form $\phi_{NS-}(\tau,\xi)$. As explained in appendix \ref{a:multiJacobi}, this multiplicity depends only on the class of $w$ in $\rootQ_g/\tilde \rootQ_g^*\cong \ZZ_2$ and on the norm $\langle \gamma|\gamma\rangle=-2mn+(w|w)$. In fact, in this case, the norm $-2mn+(w|w)$ is an even or odd integer depending on whether the class of $w$ is trivial or not in $\rootQ_g/\tilde \rootQ_g^*\cong \ZZ_2$. Thus, it is sufficient to choose a representative $w$ for each class in $\ZZ_2$, and check for which $n$ one has $c_{NS-}(n,w)\neq 0$; recall that this is the number of states of $g$-weight $w$ in $F_{24}$ with negative fermion number and $L_0-\frac{1}{2}=n$. The multiplicities of real roots correspond to $c_{NS-}(n,w)$ with $2n<(w|w)$, so there are only a finite number of states to check in order to find all real root multiplicities. 

For the trivial class in $\rootQ_g/\tilde \rootQ_g^*$, the shortest vector is $w=0$, and the lowest $n$ for which $c_{NS-}(n,0)\neq 0$ is $n=0$, with $c_{NS-}(0,0)=8$. The corresponding vectors of weight $n+\frac{1}{2}=\frac{1}{2}$ in $F_{24}$ are of the form $\lambda^i_{-1/2}|0\rangle$, $i=1,\ldots,8$, where $\lambda^i$ are the $8$ free fermions corresponding to the Cartan subalgebra of $g$.  This means that all non-zero even roots of $\g$ with zero norm have multiplicity $c_{NS-}(0,0)=8$. Furthermore, all even roots $\gamma=(m,n,w)\in \hat \Delta_0$ with $w$ in the trivial class of $\rootQ_g/\tilde \rootQ_g^*\cong \ZZ_2$ (equivalently, $\langle \gamma|\gamma\rangle\in 2\ZZ$) have norm at most $0$; in particular, there no real roots with even norm.

 For the non-trivial class in $\rootQ_g/\tilde \rootQ_g^*$, a short vector is given by $\theta_1$, and the first non-zero Fourier coefficient is $c_{NS-}(0,\theta_1)=1$, corresponding to a state  $\lambda^{\theta_1}_{-1/2}|0\rangle$, where $\lambda^{\theta_1}$ is the free fermion corresponding to the root $\theta_1$. Thus, all roots $\gamma=(m,n,w)$ with $w$ in the non-trivial class (equivalently, with $\langle \gamma|\gamma\rangle\in 2\ZZ+1$) have $\langle \gamma|\gamma\rangle\le 1$. We conclude that the even real roots of $\g$ are exactly the vectors $(m,n,w)\in \Qg_g$ with norm $1$, and their multiplicity is $c_{NS-}(0,\theta_1)=1$, as expected.
 
As for the odd roots, there are again two classes of $w$ in $(\rho+\rootQ_g)/\tilde \rootQ_g^*$, with shortest representatives $-\rho$ and $-\rho+\theta_1$, both of square length $(\rho|\rho)=(\theta_1-\rho|\theta_1-\rho)=2$. For both these representatives, the smallest $n$ for which $c_{R\pm}(n,w)\neq 0$ is $n=1$, corresponding to Ramond ground states of weight $n+\frac{1}{2}=\frac{3}{2}$, and both with multiplicity $c_{R\pm}(1,w)=2^{2/r}-1=8$ (formulas \eqref{JacRam} and \eqref{Jacodd} imply that all $c_{R\pm}(n,w)$ are multiple of $2^{r/2-1}$.) Thus, the odd roots have maximal norm $-2n+(\rho|\rho)=-2n+(\theta_1-\rho|\theta_1-\rho)=0$ and in particular there are no odd real roots. This is consistent with the observation that odd roots cannot have multiplicity $1$. The fact that the coefficients $c_{R\pm}(1,w)$ are the same for $w$ in the two classes of $$(\rho+\rootQ_g)/\tilde \rootQ_g^*$$ is not a coincidence: the coefficients $c_{R\pm}(n,w)$ are invariant under the Weyl group of $su(2)^{\oplus 8}$, and some elements in this Weyl group exchange a vector $w$ in one class of $$(\rho+\rootQ_g)/\tilde \rootQ_g^*$$ with a vector of the same norm in the other class. As a consequence, $c_{R\pm}(n,w)$ only depend on the discriminant $2n-(w|w)$; equivalently, the multiplicity of odd roots $\gamma=(m,n,w)$ only depends on their norm $\langle \gamma|\gamma\rangle=-2mn+(w|w)$.

\subsection{Weyl group}
Let us now consider the Weyl group of the BKM algebra $\g$, which is generated by reflections with respect to real roots.  As discussed above, the even root lattice $\Gamma^{1,1}\oplus \ZZ^8$ is the (unique, up to isomorphisms) odd unimodular lattice $I^{9,1}$ of signature $(9,1)$, and the real roots are all vectors of norm $1$ in this lattice. The Weyl group $W$ of $\g$ is the group of  automorphisms of $I^{9,1}$ generated by reflections with respect to norm $1$ vectors. This reflection group is studied in \cite{Borcherds1987,Borcherds1990}. As usual, one splits the set $\hat\Delta^{real}$ of real roots into the disjoint union $\hat\Delta^{real}=\hat\Delta^{real}_+\sqcup \hat\Delta^{real}_-$ of positive and negative ones, depending on the sign of the product with a regular element. There is an infinite number of simple real roots (i.e. positive roots that cannot be written as sum of other positive roots), whose corresponding reflections generate $W$. Simple roots are characterized as the vectors of norm $1$ that have inner product $1/2$ with the vector $\hat\rho=(-1,-1,\rho)$. The set of simple roots can be identified with the vectors of the affine $E_8$ lattice, in the sense that, for any choice of an arbitrary fixed simple root $x_0$, the set of vectors $x-x_0$, where $x$ is any simple root, form a copy of the $E_8$ lattice. The full group $\tilde W$ of reflection automorphisms of $I^{9,1}$ is strictly larger than $W$, and includes reflections with respect to vectors of norm $2$. It is also finitely generated, and the quotient $\tilde W/W$ is isomorphic to the affine Weyl group of $E_8$, $W^{aff}(E_8)=E_8\ltimes W(E_8)$. The group $W^{aff}(E_8)$ acts on the set of simple roots of $W$: $ W(E_8)$ is the subgroup that fixes a given simple root (say $x_0$), while the $E_8$ factor in $W^{aff}(E_8)$ acts by translations by $E_8$ lattice vectors. Since the multiplicities of  both the odd and the even roots of $\g$ only depend on their norm, they are actually invariant under the full group of automorphisms of $I^{1,9}$, and in particular under $\tilde W$.

\subsection{Root multiplicities and denominator formulas}\label{s:a18roots}
In order to find the root multiplicities for all roots of the algebra, it is more useful to adopt a different description of the $F_{24}$ SVOA with the $\CN=1$ structure corresponding to $A_1^8$. We bosonize the $8$ pairs of fermions $\lambda^{\theta_k},\lambda^{-\theta_k}$, $k=1,\ldots,8$, by replacing them by $8$ chiral free scalars $Y^1,\ldots, Y^8$ compactified on $\ZZ^8$, by setting $i\partial Y^k=:\lambda^{\theta_k},\lambda^{-\theta_k}:$ and $\lambda^{\pm\theta_k}=e^{\pm iY^k}$. The remaining $8$ fermions, corresponding to the Cartan subalgebra of $A_1^8$ are now interpreted as the superpartners of the currents $\partial Y^k$. In this description, it is easy to obtain the NS partition function 
\be \phi_{NS}(\tau,\xi)= \frac{\Theta_{\ZZ^8}(\tau,\xi)}{\eta(\tau)^8}\times \frac{\theta_3(\tau)^4}{\eta(\tau)^4}\ ,
\ee where the first factor comes from the free scalars $Y^i$ and the second from the $8$ free fermions, their superpartners. We are interested in the function $\phi_{NS-}$ counting the negative fermion number states, which is obtained from $\phi_{NS}$ by keeping only the integral powers of $q$. It is convenient to split the theta function as
\be \Theta_{\ZZ^8}(\tau,z)=\Theta_{D_8}(\tau,z)+\Theta_{v+D_8}(\tau,z)\ ,
\ee where $\Theta_{D_8}$ contains only integral powers of $q$ and $\Theta_{v+D_8}$ only the half-integral ones. Here, $D_8$ is the lattice \eqref{D8lattice}, $v+D_8$ is the translate
\be v+D_8=\{(x_1,\ldots,x_8)\in \ZZ^8\mid \sum_i x_i\in 2\ZZ+1\}\ ,
\ee of $D_8$ by $v=(1,0,\ldots,0)$, and we used $\ZZ^8=D_8\cup (v+D_8)$. We perform a similar splitting of the function $f(\tau)=\frac{\theta_3(\tau)^4}{\eta(\tau)^{12}}$, i.e.
\be f(\tau)=f_{even}(\tau)+f_{odd}(\tau)\ ,
\ee where $f_{even}$ and $f_{odd}$ contain only integral and half-integral powers of $q$, respectively. Then, we have
\be \phi_{NS-}(\tau,\xi)=\Theta_{D_8}(\tau,z)f_{even}(\tau)+\Theta_{v+D_8}(\tau,z)f_{odd}(\tau)\ .
\ee We recognize this form as the theta decomposition of the Jacobi function $\phi_{NS-}$ (see appendix \ref{a:multiJacobi}), with
\be f_{even}(\tau)=\sum_{D\in 2\ZZ} c_{NS-}(D,[0]) q^{\frac{D}{2}}\ ,\qquad\qquad f_{odd}(\tau)=\sum_{D\in 2\ZZ+1} c_{NS-}(D,[v]) q^{\frac{D}{2}}.
\ee Thus, the multiplicities of an even root $\gamma=(m,n,\sum_i k_i\theta_i)$ of $\g$ is given by  $c_{NS-}(-\langle \gamma|\gamma\rangle,[0])$ or $c_{NS-}(-\langle \gamma|\gamma\rangle,[v])$, depending on whether the norm $\langle \gamma|\gamma\rangle=-2mn+\sum_i k_i^2$ is an even or odd integer. Altogether, the multiplicities of even roots are the Fourier coefficients of the function
\be f(\tau)=f_{even}(\tau)+f_{odd}(\tau)=\frac{\theta_3(\tau)^4}{\eta(\tau)^{12}}=q^{-1/2}\prod_{n=1}^\infty \frac{(1+q^{n-1/2})^8}{(1-q^{n})^8}=\frac{\eta(\tau)^{8}}{\eta(2\tau)^8\eta(\tau/2)^8}\ .
\ee The Ramond sector of the theory is given by the product of the module of the $\ZZ^8$ SVOA corresponding to the coset $\rho+\ZZ^8$, times the Ramond sector for $8$ free fermions. The Ramond partition function, therefore, is
\be \phi_R(\tau,\xi)=\frac{\Theta_{\rho+\ZZ^8}(\tau,\xi)}{\eta(\tau)^8}\times \frac{\theta_2(\tau)^4}{\eta(\tau)^4}=(\Theta_{\rho+D_8}(\tau,\xi)+\Theta_{\rho+v+D_8}(\tau,\xi))\frac{\theta_2(\tau)^4}{\eta(\tau)^{12}}\ .
\ee The function $\phi_{R+}=\phi_{R-}$ is obtained simply by dividing $\phi_R$ by two. The form above is already a theta decomposition, so that the multiplicities of odd roots $\gamma=(m,n,\rho+\sum_i k_i\theta_i)$ are the Fourier coefficients $c_{R-}(-\langle \gamma|\gamma\rangle,[\rho])=c_{R-}(-\langle \gamma|\gamma\rangle,[\rho+v])$ of the function
\be \frac{1}{2}\frac{\theta_2(\tau)^4}{\eta(\tau)^{12}}=8\prod_{n=1}^\infty \frac{(1+q^n)^8}{(1-q^{n})^8}=8\frac{\eta(2\tau)^8}{\eta(\tau)^{16}}\ .
\ee
This analysis shows that the BKM superalgebra associated to $A_1^8$ is a superalgebra already considered in \cite{BorcherdsFake}, and discussed also in \cite{Borcherds96}, in section 2 of \cite{BorcherdsMM}, and in example 13.7 of \cite{Borcherds98}. Besides the real simple roots described above, the algebra contains imaginary simple roots corresponding to negative integer multiples of the Weyl vector $-n\hat{\rho}$, $n\in \mathbb{N}$, all of them with multiplicity $8$. The root $-n\hat{\rho}$ is even or odd depending on $n$ being even or odd. The additive side of the denominator identity, therefore, in this case reads 
\be\label{eqn:additive} \sum_{\mathrm{w}\in W} \det(\mathrm{w})  e^{-\mathrm{w}(\hat{\rho})}\prod_{n=1}^\infty (1-e^{-n\mathrm{w}(\hat{\rho})})^{(-1)^n8}\ .
\ee 

As discussed in \cite{Borcherds98}, the denominator of the BKM algebra $\g$ admits an analytic continuation to a holomorphic automorphic form for $Aut(M)$, the group of automorphisms of the lattice $M$ which is the maximal even sublattice of the odd unimodular lattice of signature $(2,10)$. The lattice $M$ has two orbits of primitive norm zero vectors, which are associated to two different expansions of the automorphic form into infinite (Borcherds) products. One of these infinite products is the denominator of the algebra $\g$ considered in this section, while the other is the denominator of the BKM superalgebra constructed in \cite{Sch1}. In \cite{Borcherds96}, this automorphic form was also interpreted as a non-vanishing function on the moduli space of Enriques surfaces.

\section{Conclusions \& Future Directions}\label{sec:discussion}

In this note we studied some properties of the $c=12$ SVOA (holomorphic SCFT) $F_{24}$ of 24 free fermions, as well as its role as the internal, ``compactification'' SCFT in a chiral superstring worldsheet theory. The latter system is a super-analogue of Borcherds' method for proving the monstrous moonshine conjectures (see also \cite{Sch1, Harrison:2018joy}). Using this construction, we produced a new family of Borcherds-Kac-Moody superalgebras, and their corresponding denominators, labeled by semisimple Lie algebras of dimension 24 and arbitrary rank. 

As with our analogous study concerning the $c=12$ Conway module $V^{f\natural}$ \cite{Harrison:2018joy}, this note should be viewed as a warm-up for producing complete (i.e. non-chiral) low-dimensional string compactifications whose internal worldsheet SCFTS are given by products $V \otimes \bar{W}$ of these $c=12$ SVOAs, see \cite{HPPV}.\footnote{Related examples which are potentially relevant for this investigation are explored in \cite{phi12,HM3}.} Such peculiar critical string vacua have proved  relevant for understanding aspects of moonshine, including the genus zero property, when the SVOAs used are moonshine modules; this was illustrated for the Monster case in \cite{PPV1, PPV2}. We also believe these vacua, viewed as machines to produce explicit BKM algebras, can serve as useful toy systems for exploring and understanding BPS-algebras.

We conclude by highlighting a few outstanding questions raised by our study:
\begin{itemize}

\item In \S \ref{s:F24fromOb}, we described how one can obtain $F_{24}$, with a choice of supercurrent, from orbifolds of $V^{fE_8}$. It would also be interesting to understand what $\CN=1$-preserving orbifolds of $V^{f\natural}$ yield $F_{24}$ with a fixed choice of superconformal structure. The non-trivial question here is to determine whether one can obtain $F_{24}$ from an orbifold of $V^{f\natural}$ by a cyclic group. These orbifolds will be relevant in studying string theoretic dualities (see \cite{PPV1} for analogous appearances of orbifolds of the Monster and Leech VOAs in a string compactification). 

\item In \S\ref{sec:denom} we determined the product sides of the denominator and super-denominator formulas associated with the super BKM $\mathfrak{g}$. In the case of $g=A_1^8$ we  showed that $\g$ coincides with a BKM superalgebra already studied by Borcherds  \cite{BorcherdsFake,Borcherds96,BorcherdsMM,Borcherds98}, who was able to determine the additive side of these formulas and explicitly describe the simple roots of the algebra (see \S\ref{sec:A1^8}). 
It would be instructive to explicitly determine all simple roots, as well as the additive sides of the denominator identities, for the remaining $\CN=1$ structures labeled by $g$.
We leave this question for future work. 

\item A single automorphic form can have distinct expansions at different cusps in moduli space; the expansions can each be (super)denominators for different BKM algebras (as in, e.g., \cite{Gritsenko2012}). When embedded into a string theory construction, the BKM algebras are expected to be associated to different perturbative descriptions of the model, and related to one another via dualities \cite{PPV1, PPV2, phi12}. As just mentioned above, in the example of $g=A_1^8$ the denominator of the BKM superalgebra arises from the expansion along the ``level 2 cusp'' of a holomorphic automorphic form $\Psi$ on $\Gamma\backslash SO(2,10)/(SO(2)\times SO(10))$, where $\Gamma = Aut(M)$ (defined below Eqn \ref{eqn:additive}), a moduli space closely related to that of the Enriques Calabi-Yau threefold. The same automorphic form $\Psi$ can also be expanded along its ``level 1 cusp'' in which case it gives rise to a denominator formula of another BKM-algebra \cite{Borcherds98}. It would be fascinating to understand if the $A_1^8$ BKM (or $F_{24}$) played a role in organizing BPS states in a string compactification on the Enriques CY, in some perturbative duality frame, and if it could be related to the BKM at the other cusp of $\Psi$ by an explicit string duality. 

\item To expand on the previous point, we further note that the same automorphic form $\Psi$ arises as the genus one topological amplitude $F_1$ in the ``geometric reduction'' of the FHSV model \cite{Klemm:2005pd}, i.e. in type II string theory on the Enriques CY $X$. In this context $\Psi$ can be interpreted as a counting rational curves, i.e. Gromov-Witten invariants, on $X$. The expression for $\Psi$, expanded along the level-1 cusp, coincides with the form of the gravitational threshold correction of the FHSV-model obtained in \cite{HM4}. In view of these observations, and their connection to BPS states in string theory on $X$, it would be interesting to further explore the role of $\mathfrak{g}$ and $F_{24}$ in this context, and to connect our BKMs to curve-counts.

\item We constructed our BKM algebras from the cohomology of ``physical states'' in our chiral construction. In a true string theory, one must take the semirelative cohomology; certain variants of this cohomology (e.g., \cite{BZ}) contain information about anomalies and D-brane states. It would be very interesting to explore these cohomologies in the corresponding non-chiral string constructions. 
\item More generally, it would be very interesting to better understand the D-brane states in the non-chiral string models and their representations under moonshine groups. See \cite{CGH} for an exploration of boundary states in a bosonic Monster string theory. 
\item Though there is not moonshine for $F_{24}$ as there is for its close cousin $V^{f\natural}$, there are numerous modular coincidences among their McKay-Thompson series. It would be fascinating to see if/how the full string theory construction detects the genus zero property for $V^{f\natural}$, particularly in contrast with the other $c=12$ SVOA compactifications. The BKMs constructed in this note should play a key role in that study. 

\item Finally, it would be interesting to study the discrete symmetry groups of our BKM algebras. Various sporadic symmetry groups have been shown to stabilize extended superconformal algebras within $V^{f\natural}$ \cite{M5}, on the one hand, and certain sub-VOAs of $V^\natural$ \cite{Bae:2020pvv} on the other. It may be interesting to explore generalizations of both of these constructions for $F_{24}$.
\end{itemize}

\section*{Acknowledgements} It is a pleasure to thank R. Borcherds, S. Carnahan and V. Gritsenko for very helpful email correspondences. 

The work of S.M.H. is supported by the National Science and Engineering Council of Canada, an FRQNT new university researchers start-up grant, and the Canada Research Chairs program.
NMP was supported by a Sherman Fairchild Postdoctoral Fellowship at Caltech and by the U.S. Department of Energy, Office of Science, Office of High Energy Physics, under Award Number de- sc0011632, and is currently supported by the grant NSF PHY-1911298, and the Sivian Fund. D.P. was supported by the Swedish Research Council (Vetenskapsr\aa det), grant nr. 2018-04760.

\appendix
\section{Multivariable Jacobi forms}\label{a:multiJacobi}

In this section we first recall some known facts about multivariable Jacobi forms, and then use them to obtain some useful results about the Fourier coefficients of the partition functions for the $F_{24}$ SVOA. We follow the treatment in \cite{Gritsenko1988,Gritsenko2012}, and refer to those articles for proofs and details.

Consider an even positive definite lattice $L$ with bilinear form $(\cdot,\cdot)$. A Jacobi form form of weight $k\in \ZZ$ and index $m\in \NN$ for $L$ is a holomorphic function $\varphi(\tau,\xi)$ on $\HH \times (L\otimes\CC)$ satisfying
\begin{align}\label{modJac}
\varphi(\frac{a\tau+b}{c\tau+d},\frac{\xi}{c\tau+d})&=(c\tau+d)^k e^{\pi i\frac{mc(\xi,\xi)}{c\tau+d}}\varphi(\tau,\xi) && \begin{pmatrix}
a & b\\ c & d
\end{pmatrix}\in SL_2(\ZZ)\ ,\\
\label{ellJac}\varphi(\tau,\xi+\lambda\tau+\mu)&=e^{-\pi i m( (\lambda,\lambda)\tau+2(\lambda,\xi))}\varphi(\tau,\xi) \ ,&& (\lambda,\mu)
\in L\times L\ .
\end{align} The Jacobi form is called weak, holomorphic, or cusp, if in its Fourier expansion
\be \varphi(\tau,\xi)= \sum_{\substack{n\in \ZZ\\ \ell\in L^*}} c(n,\ell) q^n e^{2\pi i (\xi,\ell)}\ ,\qquad q=e^{2\pi i \tau}\ .
\ee  the sum over $n$ and $\ell$ is restricted to, respectively, $n\ge 0$, or $2mn-(\ell,\ell)\ge 0$, or $2mn-(\ell,\ell)>0$. It is called weakly holomorphic if $\Delta(\tau)^N\varphi(\tau,\xi)$ is a weak Jacobi form for some $N\in \NN$, with $\Delta(\tau)=\eta(\tau)^{24}$. These definitions can be generalized in the obvious way to Jacobi forms with respect to subgroups of $SL_2(\ZZ)$. Furthermore, one can consider Jacobi forms of half-integral index, at the cost of introducing some sign in the transformation properties \eqref{modJac} and \eqref{ellJac}.

The condition \eqref{ellJac} implies that the coefficients $c(n,\ell)$ only depend on $2mn-(\ell,\ell)$ and on the image of $\ell$ in $L^*/mL$.

According to this definition, `ordinary' single-variable Jacobi forms of weight $k$ and index $m$ as defined, for example, in \cite{EichlerZagier} are Jacobi forms of the same index and weight for the $1$-dimensional even lattice $L=\sqrt{2}\ZZ$. The Jacobi theta functions
\begin{align}
&\theta_1(\tau,z)=-\theta\left[\begin{smallmatrix}
\frac12\\ \frac12
\end{smallmatrix}\right]=-iq^{\frac{1}{8}} (y^{\frac{1}{2}}-y^{-\frac{1}{2}})\prod_{n=1}^\infty (1-q^n)(1-q^ny)(1-q^ny^{-1})=-\sum_{n\in \ZZ} q^{\frac{1}{2}(n+\frac{1}{2})^2} e^{2\pi i (n+\frac{1}{2})(z+\frac{1}{2})}\\
&\theta_2(\tau,z)=\theta\left[\begin{smallmatrix}
\frac12 \\ 0
\end{smallmatrix}\right]=q^{\frac{1}{8}}  (y^{\frac{1}{2}}+y^{-\frac{1}{2}})\prod_{n=1}^\infty (1-q^n)(1+q^ny)(1+q^ny^{-1})=\sum_{n\in \ZZ} q^{\frac{1}{2}(n+\frac{1}{2})^2} e^{2\pi i (n+\frac{1}{2})z}\\
&\theta_3(\tau,z)=\theta\left[\begin{smallmatrix}
0 \\ 0
\end{smallmatrix}\right]=\prod_{n=1}^\infty (1-q^n)(1+q^{n-\frac{1}{2}}y)(1+q^{n-\frac{1}{2}}y^{-1})=\sum_{n\in \ZZ} q^{\frac{1}{2}n^2} e^{2\pi i nz}\\
&\theta_4(\tau,z)=\theta\left[\begin{smallmatrix}
0 \\ \frac{1}{2}
\end{smallmatrix}\right]=\prod_{n=1}^\infty (1-q^n)(1-q^{n-\frac{1}{2}}y)(1-q^{n-\frac{1}{2}}y^{-1})=\sum_{n\in \ZZ} q^{\frac{1}{2}n^2} e^{2\pi i n(z+\frac{1}{2})}
\end{align}
 are Jacobi forms of weight $1/2$ and index $1/2$ for a subgroup of index $3$ in $SL_2(\ZZ)$.

Let us now show that the functions $\phi_X(\tau,\xi)$, $X\in \{NS,\tilde{NS},R,\tilde{R},NS\pm,R\pm\}$, defined in section \ref{s:partfunct} are weakly holomorphic Jacobi forms of index $m=1$ and weight $0$ for the  lattice $L=\rootQ_g^\vee$ (the coroot lattice of the algebra $g$). Let us first notice that all such functions are given by a product $\prod_{\alpha\in\Delta_+} \theta_i(\tau,(\xi|\alpha))$ of theta functions times modular function that depends on $\tau$ only.

Given the elliptic properties of the theta functions
\be \theta\left[\begin{smallmatrix}
	a \\ b
\end{smallmatrix}\right](\tau,z+n+m\tau)=(-1)^{2an+2bm} e^{-\pi i (m^2\tau+2mz)}\theta\left[\begin{smallmatrix}
a \\ b
\end{smallmatrix}\right](\tau,z)\ ,\qquad \qquad n,m\in \ZZ\ ,
\ee where $a,b\in \{0,\frac{1}{2}\}$,
we get, for all $\lambda,\mu\in \rootQ^\vee_g$ (coroot lattice)
\begin{align} &\prod_{\alpha\in \Delta^+}\theta_i(\tau,(\xi+\lambda\tau+\mu|\alpha))=\prod_{\alpha\in \Delta^+}\theta_i(\tau,(\xi|\alpha)+(\lambda|\alpha)\tau+(\mu|\alpha))\\&=\pm e^{-\pi i \sum_{\alpha\in\Delta+}((\lambda|\alpha)^2\tau+2(\lambda|\alpha)(\alpha|\xi))}\prod_{\alpha\in \Delta^+}\theta_i(\tau,(\xi|\alpha))
\\&=\pm e^{-\pi i ((\lambda|\lambda)\tau+2(\lambda|\xi))}\prod_{\alpha\in \Delta^+}\theta_i(\tau,(\xi|\alpha))\ .
\end{align}  for $i=2,3,4$,
where we used the identity\footnote{Once again, the relative normalization of the two sides of this identity depends on our choice \eqref{algnorm} of Killing form.}
\be \sum_{\alpha\in\Delta+} (\lambda|\alpha)(\alpha|\mu)=(\lambda|\mu)\qquad \forall \lambda,\mu\in \rootQ_g\otimes\RR\ .
\ee This implies that $\prod_{\alpha\in \Delta^+}\theta_i(\tau,(\xi|\alpha))$ has the elliptic properties of a Jacobi form of index $1$ for the even lattice $\rootQ^\vee_g$. The general theory of Jacobi forms for lattices tells us that the Jacobi forms $\phi_X(\tau,\xi)$ admits a Fourier expansion of the form \eqref{JacobiFourier}, where the coefficients $c_X(n,w)$ only depend on
\be D\equiv D(n,w)=2n-(w,w)\ ,\ee 
and on the class $[w]$ of $w$ in the quotient $P_g/i(\rootQ^\vee_g)$, possibly up to a sign. We will sometimes use the notation $c(D,[w])$ to stress this dependence.  The sign is easily recovered by noticing that, by definition, all Fourier coefficients $c_X(n,w)$ are non negative, except when $X=\tilde{NS}$, where the sign is $(-1)^{2n}$, with $n\in \frac{1}{2}\ZZ$.

When $X=NS-$ or $R\pm$, the sum over $n$ in the Fourier expansion is bounded by $n\ge 0$. Therefore, if $c(D,[w])\neq 0$ then  for all $w'\in w+\rootQ^\vee_g$ we must have $2n\equiv D+(w'|w')\ge 0$. If $m([w])$ is the minimal squared length of a vector in the coset $w+\rootQ^\vee_g$, we get the bound
\be\label{coeffbound1} c(D,[w])\neq 0\qquad\Rightarrow\qquad D\ge -m([w])\ .
\ee This bound can be also written as
\be\label{coeffbound2} c(n,w)\neq 0\qquad \Rightarrow\qquad ( w|w) \le 2n+m([w]) \ ,
\ee which shows that for each given $n$ there are only a finite number of vectors $w\in P_g$ for which $c(n,w)\neq 0$.

The fact that the coefficients $c(D,[w])$ only depend on the discriminant $D$ and on $[w]\in P_g/i(\rootQ^\vee_g)$ implies that the Jacobi functions admit a theta decomposition
\begin{align} \phi_X(\tau,\xi)&=\sum_{[w]\in P_g/i(\rootQ^\vee_g)} \sum_D \sum_{w'\in w+\rootQ^\vee_g} c_X(D,[w]) q^{\frac{D+(w'|w')}{2}} e^{2\pi i(\xi|w')}\\
\label{thetadec}&=\sum_{[w]\in P_g/i(\rootQ^\vee_g)} h_{X,[w]}(\tau)\Theta_{w+\rootQ^\vee_g}(\tau,\xi)\ ,
\end{align} where
\be \Theta_{w+\rootQ^\vee_g}(\tau,\xi)=\sum_{w'\in w+\rootQ^\vee_g} q^{\frac{(w'|w')}{2}} e^{2\pi i(\xi|w')}
\ee is the theta series of the coset $w+\rootQ^\vee_g$, and
\be h_{X,[w]}(\tau)=\sum_{D} c_X(D,[w]) q^{D/2}\ ,
\ee is a weakly holomorphic modular form containing all non-trivial information about the Fourier coefficients $c_X$.

In some cases, the functions $\phi_X(\tau,z)$ are Jacobi forms with respect to a lattice that is `finer' than $\rootQ^\vee$, and this leads to more stringent conditions on their Fourier coefficients. In particular, the coefficients $c_{NS}(n,w)$ and $c_{R}(n,w)$ are nonzero only for $w\in \rootQ_g\subseteq P_g$ and $w\in \rho+\rootQ_g\subseteq P_g$, respectively. If the lattice $\tilde Q_g$ generated by $\rho$ and $\rootQ_g$ is a proper sublattice of the weight lattice $P_g$, then $\phi_{NS}(\tau,z)$ and $\phi_{R}(\tau,z)$ are Jacobi forms with respect to any lattice $\tilde Q_g^\vee$ that is even and contained in the dual $(\tilde Q_g)^*$, so that $\tilde Q^\vee_g\supseteq \rootQ^\vee$. This means that $c_X(n,w)$ only depends on the discriminant $2n-(w|w)$ and on the coset of $w+\tilde Q^\vee_g$, rather than $w+ \rootQ^\vee$. Correspondingly, the theta decomposition \eqref{thetadec} becomes
\be \phi_X(\tau,\xi)=\sum_{[w]\in \tilde Q_g/\tilde Q^\vee_g} h_{X,[w]}(\tau)\Theta_{w+\tilde Q^\vee_g}(\tau,\xi)\ .
\ee

For example, when $g=(A_1)^{\oplus 8}$, one has $Q_g=\ZZ^{\oplus 8}$, $P_g=(\frac{1}{2}\ZZ)^{\oplus 8}$ and $\rootQ^\vee_g=(2\ZZ)^{\oplus 8}$, with $\rho=(\frac{1}{2},\ldots,\frac{1}{2})\in P_g$. In this case, the lattice $\tilde Q_g=\rootQ_g\cup (\rho+\rootQ_g)$ is given by $\ZZ^{\oplus 8}\cup (\frac{1}{2}+\ZZ)^{\oplus 8}$. The dual of $\tilde Q_g$ is
\be \tilde Q_g^*=\{(x_1,\ldots,x_8)\in \ZZ^{\oplus 8}\mid \sum_i x_i\in 2\ZZ\}\ ,
\ee which is an even lattice (isomorphic to the $D_8$ lattice), so that we can set $\tilde Q_g^\vee:=\tilde Q_g^*$. Thus, $c_{NS}(n,w)$ and $c_{R}(n,w)$ depend only on $2n-(w|w)$ and on the class $[w]\in \rootQ_g/\tilde Q^\vee_g\cong \ZZ_2$ (NS sector) or $[w]\in (\rho+\rootQ_g)/\tilde Q^\vee_g$ (R sector). By comparison, one has $Q_g/\rootQ_g^\vee=\ZZ_2^{8}$, so that, just for the NS sector, using the most naive constraints one needs to compute the coefficients for $2^8$ different classes rather than just $2$.

\section{Details about cohomology}\label{a:details}

In this note, we follow a chiral version of the construction of the relative cohomology of physical string states. The BRST charge is given by
\begin{align} Q=&\sum_m c_mL^m_{-m}+\sum_r \gamma_r G^m_{-r} \\&+\sum_{m,n} \frac{1}{2}(n-m) : b_{-n-m} c_nc_m: +\sum_{m,r}[ \frac{1}{2} (2r-m) :\beta_{-m-r}c_m\gamma_r: -:b_{-m}\gamma_{m-r}\gamma_r:]+ac_0
\end{align} where $a=-\frac{1}{2}$ in the NS sectors and $a=-\frac{5}{8}$ in the Ramond sector. Relative cohomology (which is equivalent to the physically relevant semirelative cohomology for a non-chiral theory) is given by considering $Q$-closed states in the kernel of $b_0$, modulo states of the form $|\chi \rangle \sim Q|\lambda \rangle$ with $|\lambda \rangle$ in $\ker b_0$. 

As discussed in the main text, the cohomology classes for zero momentum have to be treated separately from the nonzero momentum states, but are amenable to a direct computation using standard techniques and explicit representatives.

The zero momentum states in the $-1$-picture with $L_0=0$ are obtained by acting by any operator of weight $1/2$ on the ground state $e^{-\phi}c_1|0\rangle$. States with integral $L_0$ eigenvalue are automatically included by the GSO projection. There are the following possibilities:
\begin{itemize}
	\item Suppose the internal SVOA $V$ has $N$ states $v^a$, $a=1,\ldots, N$, of weight $1/2$ in the NS sector.  Then, we have $N$ states
	$$v^a_{-1/2}e^{-\phi}c_1|0\rangle\ ,\qquad i=1,\ldots,N,$$
	with ghost number $1$.
	\item There are two states
	$$\psi^\mu_{-1/2}e^{-\phi}c_1|0\rangle\ ,\qquad\qquad \mu\in\{+,-\}$$
	again with ghost number $1$.
	\item One state
	$$ \gamma_{-1/2}e^{-\phi}c_1|0\rangle
	$$ with ghost number $2$.
	\item One state
	$$ \beta_{-1/2}e^{-\phi}c_1|0\rangle
	$$ with ghost number $0$.
\end{itemize} Notice that
\be \{Q,c_n\}= \sum_m \frac{1}{2}(n-2m):c_{n-m}c_m:-\sum_r :\gamma_{n-r}\gamma_r:
\ee and in particular
\be \{Q,c_1\}= \sum_{m>0} (1-2m)c_{1-m}c_m-\sum_r :\gamma_{1-r}\gamma_r:
\ee
With non-zero null momentum $k$, $k^2=0$, the BRST variation of $\beta_{-1/2}e^{-\phi}c_1|k\rangle$ is proportional to $k_\mu\psi^\mu_{-1/2}e^{-\phi}c_1|k\rangle$, while the BRST variation of $\psi^\mu_{-1/2}e^{-\phi}c_1|k\rangle$ is proportional to $k^\mu \gamma_{-1/2}e^{-\phi}c_1|k\rangle$; $v^a_{-1/2}e^{-\phi}c_1|k\rangle$ and $\gamma_{-1/2}e^{-\phi}c_1|k\rangle$ are always $Q$-closed (the latter is obvious, since there are no states with ghost number $3$). Therefore, when $k\neq 0$, we have $N+2$ closed states ($v^a_{-1/2}e^{-\phi}c_1|k\rangle$, $a=1,\ldots, N$,  $\gamma_{-1/2}e^{-\phi}c_1|k\rangle$ and one linear combination of $\psi^\mu_{-1/2}e^{-\phi}c_1|k\rangle$), but two of them are $Q$-exact ($\gamma_{-1/2}e^{-\phi}c_1|k\rangle$ and the linear combination of $\psi^\mu_{-1/2}e^{-\phi}c_1|k\rangle$), so we are left with $N$ classes in $H^1(k)_{p=-1}$ with representatives $v^a_{-1/2}e^{-\phi}c_1|k\rangle$.  When $k=0$, all these states are $Q$-closed, and they therefore correspond to distinct cohomology classes. The dimensions of the non-zero cohomology spaces are therefore
\begin{align*}
\dim H^0(k=0)_{p=-1}=1\ ,\\\dim H^1(k=0)_{p=-1}=N+2\ ,\\ \dim H^2(k=0)_{p=-1}=1\ .
\end{align*}

Let us now consider the Ramond sector. Let us assume that the SVOA $V$ has $K_+$  (respectively, $K_-$) Ramond states $u^{i+}$, $i=1,\ldots, K_+$ (respectively, $u^{i-}$, $i=1,\ldots, K_-$) with weight $1/2$ and $V$-fermion number $(-1)^{F_V}$ equal to $+1$ (respectively, $-1$). The Ramond sector of the $V^{X,\psi}$ `space-time' vertex algebra contains two ground states $|k,\pm\rangle$ with momentum $k$ where the sign denotes space-time spin (and the fermion number). Then, in the $(-1/2)$-picture, the $k=0$ states with total fermion number $(-1)^{F_{tot}}=+1$ and in $\ker b_0\cap \ker\beta_0$ are:
\begin{itemize}
	\item $e^{-\phi/2}c_1|0,u^{i+},+\rangle$, $i=1,\ldots,K_+$ 
	\item $e^{-\phi/2}c_1|0,u^{i-},-\rangle$, $i=1,\ldots,K_-$ 
\end{itemize} all of them with ghost number $1$. If we drop the requirement that the states are in $\ker\beta_0$, then we have states
\be\gamma_0^{n-1}ne^{-\phi/2}c_1|0,u^{i+},(-1)^{n-1}\rangle, i=1,\ldots,K_+\ee
\be\gamma_0^{n-1}e^{-\phi/2}c_1|0,u^{i-},(-1)^{n}\rangle, i=1,\ldots,K_-\ee
for each ghost number $n\ge 1$. There are no states with ghost number $n\le 0$.

For ghost number $1$, one has
\be Qe^{-\phi/2}c_1|0,u^{i\pm},\pm\rangle=\gamma_0G_0^me^{-\phi/2}c_1|0,u^{i\pm},\pm\rangle\ee
Actually, for all the matter SVOA we are considering, the Ramond ground states are all contained in $\ker G_0^m$, so that
all the states are $Q$-closed and represent $K^++K^-$ distinct cohomology classes (since there are no states with ghost number $0$, there cannot be exact states at ghost number $1$).

At higher ghost number, we use
\be [Q,\gamma_n]=\sum_r\frac{1}{2}(3r-n):c_{n-r}\gamma_r:
\ee and in particular
\be [Q,\gamma_0]=\sum_r\frac{3}{2}r:c_{-r}\gamma_r:
\ee
to conclude that
\be
Q\gamma_0^{n-1}e^{-\phi/2}c_1|0,u^{i\pm},\pm(-1)^{n-1}\rangle=\gamma_0^{n-1}Qe^{-\phi/2}c_1|0,u^{i\pm},\pm(-1)^{n-1}\rangle=0\ .
\ee Thus, all  cohomology groups of degree $n\ge 1$ are isomorphic to each other, with the isomorphism given  by $\gamma_0$.

In the $-3/2$-picture,  the $k=0$ states with total fermion number $(-1)^{F_{tot}}=+1$ and in $\ker b_0$ are:
\begin{itemize}
	\item $\beta_0^nu^{i+}_{-1/2}e^{-3\phi/2}c_1|0,-\rangle$, $i=1,\ldots,K_+$ 
	\item $\beta_0^nu^{i-}_{-1/2}e^{-3\phi/2}c_1|0,+\rangle$, $i=1,\ldots,K_-$ 
\end{itemize} for all $n\ge 0$ (note that they have the opposite space-time spin, because $e^{-3\phi/2}$ and $e^{-\phi/2}$ have opposite fermion number). These states have ghost number $1-n$.

\end{document}